\documentclass[a4paper,journal,twoside]{IEEEtran}
\IEEEoverridecommandlockouts

\newcommand\hmmax{0}
\newcommand\bmmax{0}

\usepackage[ruled,vlined,linesnumbered]{algorithm2e}
\SetKwIF{If}{ElseIf}{Else}{if}{}{else if}{else}{end if}%

\usepackage{graphicx}
\usepackage{tikz}
\usepackage{pgfplots}
\usepackage{xcolor}
\usepackage{color, colortbl}
\usepackage{amsmath}
\usepackage{bbm}
\usepackage{amsfonts, amssymb}
\usepackage{mathtools}
\usepackage{enumerate}
\usepackage{cite}
\usepackage[nolist]{acronym}
\usepackage{amsmath}

\usepackage{booktabs} 		
\usepackage{diagbox}		
\usepackage{upgreek}
\usepackage{bbm}
\usepackage{Files/bm}
\usepackage{Files/winsnotation}
\usepackage{Files/Symbols_OPL}
\usepackage{Files/LVcolours} 

\newcommand{\rev}[1]{#1}

\usepackage{setspace}	 	
\usepackage{comment}
\usepackage{physics}
\usepackage{enumitem}

\newlist{myitemize}{itemize}{3}
\setlist[myitemize,1]{label=1.,leftmargin=1em}
\setlist[myitemize,2]{label=$\rightarrow$,leftmargin=0.75em}
\setlist[myitemize,3]{label=$\diamond$}

\usepackage{array}
\newcolumntype{C}[1]{>{\centering\arraybackslash}p{#1}}
\newcolumntype{l}[1]{>{\raggedright\arraybackslash}p{#1}}

\makeatletter
\def\endthebibliography{%
  \def\@noitemerr{\@latex@warning{Empty `thebibliography' environment}}%
  \endlist
}
\makeatother
\usepackage{amsthm}
\theoremstyle{definition}
\newtheorem{definition}{Definition}

\newtheorem{theorem}{Theorem}

\usepackage{algpseudocode,amsmath}

\pgfplotsset{compat=1.17}

\begin{document}

\title{Fault-Tolerant Cut-Cat State \\ Syndrome Extraction  for Quantum Codes}



\author{Diego Forlivesi,~\IEEEmembership{Graduate~Student~Member,~IEEE,}
Lorenzo~Valentini,~\IEEEmembership{Member,~IEEE,}
        and~Marco~Chiani,~\IEEEmembership{Fellow,~IEEE}
\thanks{The authors are with the Department of Electrical, Electronic, and Information Engineering ``Guglielmo Marconi'' and CNIT/WiLab, University of Bologna, V.le Risorgimento 2, 40136 Bologna, Italy. E-mail: \{diego.forlivesi2, lorenzo.valentini13, marco.chiani\}@unibo.it. 
Work funded in part by the European Union - Next Generation EU, PNRR project PRIN n. 2022JES5S2.
}
}

\maketitle 
\markboth{}{Forlivesi, Valentini, Chiani}

\begin{acronym}
\small
\acro{BCH}{Bose–Chaudhuri–Hocquenghem}
\acro{BC}{bubble clustering}
\acro{CDF}{cumulative distribution function}
\acro{CRC}{cyclic redundancy code}
\acro{LDPC}{low-density parity-check}
\acro{LUT}{lookup table}
\acro{ML}{maximum likelihood}
\acro{MWPM}{minimum weight perfect matching}
\acro{QECC}{quantum error correcting code}
\acro{PDF}{probability density function}
\acro{PMF}{probability mass function}
\acro{MPS}{matrix product state}
\acro{WEP}{weight enumerator polynomial}
\acro{WE}{weight enumerator}
\acro{BD}{bounded distance}
\acro{QLDPC}{quantum low density parity check}
\acro{CSS}{Calderbank, Shor, and Steane}
\acro{MST}{minimum spanning tree}
\acro{PruST}{pruned spanning tree}
\acro{RFire}{Rapid-Fire}
\acro{UF}{union-find}
\acro{LEMON}{library for efficient modeling and optimization in networks}
\acro{STM}{spanning tree matching}
\acro{i.i.d.}{independent identically distributed}
\acro{QEC}{quantum error correction}
\acro{BP}{belief propagation}
\acro{FT}{fault-tolerant}
\acro{LUT}{look-up table}

\end{acronym}
\setcounter{page}{1}

\begin{abstract}
Reliable quantum computation requires fault-tolerant protocols to prevent errors from propagating during syndrome extraction in quantum error correction.
We present a novel fault-tolerant syndrome extraction technique for \acs{CSS} codes, which we refer to as the cut-cat state scheme.
While each ancilla qubit interacts non-fault-tolerantly with a pair of data qubits, we introduce additional cat stabilizer measurements to identify and correct the resulting hook errors.
Our approach maintains the key benefit of cat-based extraction, i.e., parallelized data qubit interactions, while reducing the number of simultaneous qubits required by more than half.
Compared to flag-based state-of-the-art protocols, the cut-cat scheme offers a notable advantage in terms of two-qubit gate count as the code distance increases.

\end{abstract}


\section{Introduction}

One of the central challenges in quantum computing architectures is mitigating the impact of noise arising from unavoidable interactions between quantum systems and their environments~\cite{Sho:95,Laf:96,Kni:97,FleShoWin:08}.
To protect quantum information from such errors, \ac{QEC} encodes logical qubits into larger assemblies of physical qubits, using ancillary qubits to detect and correct errors without violating the no-cloning theorem.
\ac{QEC} is therefore a fundamental component for enabling \ac{FT} quantum computation, long-lived quantum memories, and the scalable execution of quantum algorithms~\cite{AhaBen97:FTschemeFp, FowMarMar:12, Ter:15, DenKitLan:02, Bab:19, ForValChi:23}.
A major class of quantum error-correcting codes is that of stabilizer codes\cite{Got:09}, among which \ac{CSS} codes are particularly well-suited for practical realization\cite{CalSho:96, Ste:96, DenKitLan:02, Bom06:colorCodes, ValForChi25:CylMob, ForValChi24:MacW}.
Indeed, \ac{CSS} codes feature transversal encoded gates by construction and support a variety of tailored \ac{FT} protocols \cite{Got:96, Ste97:SteaneGadget}.
From a decoding perspective, they also offer a notable advantage: since each stabilizer generator contains only $\M{X}$ or $\M{Z}$ Pauli operators, $\M{X}$ and $\M{Z}$ errors can be decoded independently, significantly reducing the overall decoder complexity.
As a result, several optimized decoders have been developed for different classes of \ac{CSS} codes~\cite{Hig:23, HigGid:23, Del:21, rof20:BPplusOSD, ForValChi25:Bub, ValForChi25:latency}.

In \ac{QEC}, errors are diagnosed by measuring the code generators.
However, the measurement circuits themselves are imperfect, and faults can propagate to multiple data qubits during syndrome extraction.
Such correlated faults are commonly referred to as hook errors~\cite{DenKitLan:02}.
One \ac{FT} approach involves preparing additional ancilla qubits in a special entangled state known as a cat state, which are then coupled to the data qubits to extract the syndrome~\cite{Sho96:CatState}.
Subsequent improvements reduced the ancilla overhead by coupling each ancilla qubit to two data qubits instead of one~\cite{Ste:14, Yod:17, cha18:FirstFLagFT}.
However, while the method described in~\cite{Sho96:CatState} can tolerate an arbitrary number of faults, provided that the cat state is prepared with sufficient fault-tolerance, these later schemes are only \ac{FT} up to distance three (i.e., can manage one error).
An alternative \ac{FT} approach to syndrome extraction involves the use of flag qubits, which are auxiliary qubits entangled with the syndrome qubit such that problematic hook errors also affect the flag qubits \cite{cha18:FirstFLagFT, Cha18:FlagsDet, cha19:MagicStateFlags}.
By analyzing the resulting flag pattern, the location of faults can be inferred, allowing for appropriate correction or mitigation of the induced data qubit errors.

In this paper, we introduce a novel \ac{FT} syndrome extraction technique for \ac{CSS} codes, with explicit constructions demonstrated up to distance nine, based on a modified cat state construction referred to as the cut-cat state scheme.
In this approach, each cat qubit interacts with a pair of data qubits in a non-\ac{FT} manner.
To mitigate the resulting hook errors, we perform additional cat stabilizer measurements designed to detect and correct such faults.
The cut-cat state scheme preserves the main advantage of cat-based extraction, which is the parallelized interaction with data qubits. At the same time, it reduces the number of simultaneous qubits by more than half.
The scheme is primarily tailored to CSS code families that include stabilizer generators of medium to high weight, as arise in many nonlocal or partially nonlocal constructions, including certain QLDPC~\cite{Breebe21:QLDPCcodes}, triorthogonal codes~\cite{BraHaa:12}, and concatenated quantum codes~\cite{Yam24:ConcatenatedCodes, Got24:ManyHypercube}, where the concatenation process leads to stabilizer generators whose weight grows across levels.
As a representative application, we analyze the $[[49,1,5]]$ triorthogonal quantum code \cite{TransversalClifford25:ShuVic}, which features stabilizer generators with large support and whose triorthogonal structure enables transversal non-Clifford logical gates, making it particularly well suited for \ac{FT} magic-state distillation \cite{BraHaa:12, DagBluBro:25}.
Moreover, we provide a proof of fault-tolerance of the proposed scheme for such code distances, and compare it against state-of-the-art methods based on full cat states and flag qubits, evaluating simultaneous qubit requirements, two-qubit gate count, and circuit depth.

This paper is organized as follows. Section~\ref{sec:preliminary} introduces preliminary concepts about \acp{QECC} and \ac{FT} syndrome extraction protocols. 
In Section~\ref{sec:Cat?FT}, we thoroughly describe the proposed cut-cat state protocol. 
Section~\ref{sec:Decoding} describes the designed decoding technique.
\rev{Section~\ref{sec:Proofs}} provides a proof about the fault-tolerance of the cat cut syndrome extraction method.
Finally, numerical results are presented in Section~\ref{sec:NumRes}.

\section{Preliminaries and Background}
\label{sec:preliminary}
\subsection{Stabilizer Formalism}
\label{subsec:QEC}

The Pauli operators are denoted by $\M{X}$, $\M{Y}$, and $\M{Z}$. A \ac{QECC} encoding $k$ logical qubits $\ket{\varphi}$ into a codeword of $n$ physical qubits $\ket{\psi}$ with minimum distance $d$ is denoted as $[[n, k, d]]$. Such a code can correct all error patterns affecting up to $t = \lfloor (d - 1)/2 \rfloor$ physical qubits.
In the stabilizer formalism, a code is specified by $n - k$ independent, commuting operators $\M{G}_i \in \mathcal{G}_n$, known as stabilizer generators, where $\mathcal{G}_n$ is the $n$-qubit Pauli group \cite{Got:09, NieChu:10}. The subgroup of $\mathcal{G}_n$ generated by all combinations of the $\M{G}_i$ forms the stabilizer group $\mathcal{S}$. The code space $\mathcal{C}$ consists of all quantum states $\ket{\psi}$ stabilized by $\mathcal{S}$, i.e., satisfying $\M{G}_i \ket{\psi} = \ket{\psi}$ for all $i = 1, \dots, n - k$.
Let us define as $\gamma_i$ the weight of the $i$-th generator.
Operators that commute with all elements of $\mathcal{S}$ but are not in $\mathcal{S}$ are called logical operators. The stabilizer generators define measurements that preserve the encoded quantum state and are implemented using ancillary qubits. When an error $\M{E} \in \mathcal{G}_n$ acts on a codeword, resulting in the corrupted state $\M{E}\ket{\psi}$, a binary error syndrome $\V{s}$ is extracted, where $s_i = 0$ if $\M{G}_i$ commutes with $\M{E}$ and $s_i = 1$ if they anticommute. This syndrome provides the information needed to identify and correct the error.

An important class of stabilizer codes is the family of \ac{CSS} codes \cite{CalSho:96,Ste:96}, in which generators consist solely of Pauli $\M{X}$ operators or solely of Pauli $\M{Z}$ operators. 
These codes are of particular importance for practical implementations.
Indeed, the structure of CSS codes allows for the transversal application of a CNOT gate between two code blocks to be a fault-tolerant operation, effectively performing a logical CNOT between the corresponding encoded qubits.
Furthermore, for self-dual CSS codes (for which the $\M{X}$ and $\M{Z}$ generators coincide), applying a transversal Hadamard gate implements a logical Hadamard on the encoded qubits.
Moreover, since each stabilizer generator consists solely of either $\M{X}$ or $\M{Z}$ Pauli operators, $\M{X}$ and $\M{Z}$ errors can be decoded independently, which greatly reduces the overall complexity of the decoder.

\subsection{Flag based Syndrome Extraction }

Whenever a quantum operation is performed, such as a single-qubit gate, a two-qubit gate, a measurement, or a state preparation, a fault may occur due to physical imperfections, modeled as a Pauli error.

In the following, we adopt the definition of a \ac{FT} error-correction gadget
from~\cite{Got:09}. 
For clarity, since this work focuses on the
syndrome-extraction quantum circuit, we refer to it as a \ac{FT} syndrome-extraction gadget.

\begin{definition}[FT syndrome extraction gadget]
\label{def:FTEC}
A syndrome extraction procedure for a quantum code of distance $d$ is said to be \ac{FT} if it satisfies the following conditions.
If at most $s \le t = \lfloor (d-1)/2 \rfloor$ faults occur during the procedure and no additional faults occur afterward, then the output state differs from a valid encoded state by a data error of weight at most $s$.
In addition, if the input encoded state contains a data error of weight $r$ and at most $s$ additional faults occur during the procedure, with $r + s \le t$, then an ideal decoder successfully corrects the resulting state.
\end{definition}

Syndrome extraction schemes based on directly measuring each stabilizer of the code are, in general, not \ac{FT}.
Indeed, errors arising from two-qubit gates, particularly those occurring mid-circuit, can propagate to higher-weight errors on the data qubits, as illustrated in red in Fig.~\ref{Fig:Flags}.
Such errors are commonly known as hook errors.

One way to address this issue is through the use of flag qubits: additional qubits entangled with the syndrome qubit such that hook errors are also propagated to the flag qubits \cite{cha18:FirstFLagFT, cha20:FLagAllStabCodes}.
By observing the resulting flag pattern, one can approximately localize the fault and correct the resulting data qubit errors accordingly.
An example of flag-based syndrome extraction scheme, for measuring an operator acting on $\gamma = 4$ qubits, is shown in Fig.~\ref{Fig:Flags}.
In this circuit, an $\M{X}$ error occurring after the second CNOT gate interacting with the data qubits propagates to a weight $w = 2$ error on the data.
However, the same error is also transferred to a flag qubit.
Once flag qubits are measured, a suitable correction can be applied such that no fault ever results in a higher-weight error (up to a stabilizer).
In the figure, the correction operation shown in blue mitigates the errors in the data qubits, rendering the procedure fault-tolerant. 
Specifically, a single error (in red) occurring in the ancilla qubits propagates to only one data qubit.

Moreover, measurement errors on the syndrome qubits may lead to an incorrect syndrome and the consequent application of an erroneous correction operator pattern, resulting in a non-\ac{FT} syndrome extraction gadget.
This issue is typically addressed by repeating the syndrome extraction measurement until $t+1$ identical syndromes are obtained consecutively~\cite{Sho96:CatState, TanBal23:t+1RepeatedSyndromeExraction}.
In this way, if $w \leq t$ measurement errors occur during the gadget, the final syndrome is guaranteed to be the correct one.
Taken together, these two solutions yield a \ac{FT} syndrome extraction gadget, according to Definition~\ref{def:FTEC}.

\begin{figure}[t]
 	\centering
 	\resizebox{0.7\columnwidth}{!}{
 	   \includegraphics{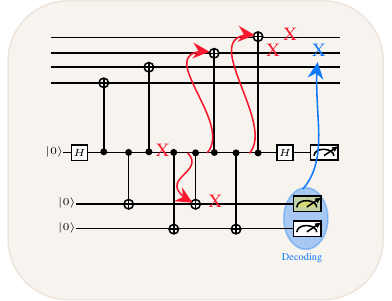}
     }
 	\caption{Flag $t$-\ac{FT} syndrome extraction: A single $\M{X}$ error on the syndrome qubit propagates to a weight $w = 2$ error on the data qubits.
    This error is also transferred to a flag qubit, enabling it to be detected and corrected.}
 	\label{Fig:Flags}
\end{figure}

\subsection{Cat State Based Syndrome Extraction }
\label{subsec:CatFTSE}

An alternative scheme for a \ac{FT} syndrome extraction gadget employs cat states, also known as GHZ states, of the form $\ket{00\dots0} + \ket{11\dots1}$ \cite{Sho96:CatState}.
The cat state can be prepared fault-tolerantly by measuring its $\M{X}$-type stabilizer generator using, for example, flag-based gadgets.
Alternatively, it can be prepared via post-selection methods~\cite{ForAma25:FTCSS}.
Each qubit of the cat state then interacts with a corresponding data qubit via a two-qubit gate, with the cat qubit acting as control, according to the operator being measured.
Finally, a transversal Hadamard gate followed by measurement in the computational basis is applied to the cat state.
By examining the parity of the measurement outcomes, the eigenvalue of the measured operator can be determined.
In this way, any single error occurring on a cat-state qubit propagates to at most a single data qubit.
By repeating the syndrome extraction measurement until $t+1$ identical syndromes are obtained consecutively, the resulting syndrome extraction gadget is \ac{FT}.

The main advantage of a flag-based syndrome extraction scheme is that it requires fewer qubits to be active simultaneously.
In particular, it uses only a single syndrome qubit, unlike cat-state schemes, which require as many qubits as the weight of the operator being measured.
On the other hand, a key advantage of using cat states is that they can be prepared offline, potentially with the aid of detection circuits that require fewer flag qubits.
Moreover, the portion of the circuit that interacts with the data qubits can be executed fully in parallel, which helps mitigate memory errors.

\section{Cut-Cat State FT Syndrome Extraction}
\label{sec:Cat?FT}

\begin{figure*}[t]
 	\centering
 	\resizebox{\textwidth}{!}{ \includegraphics{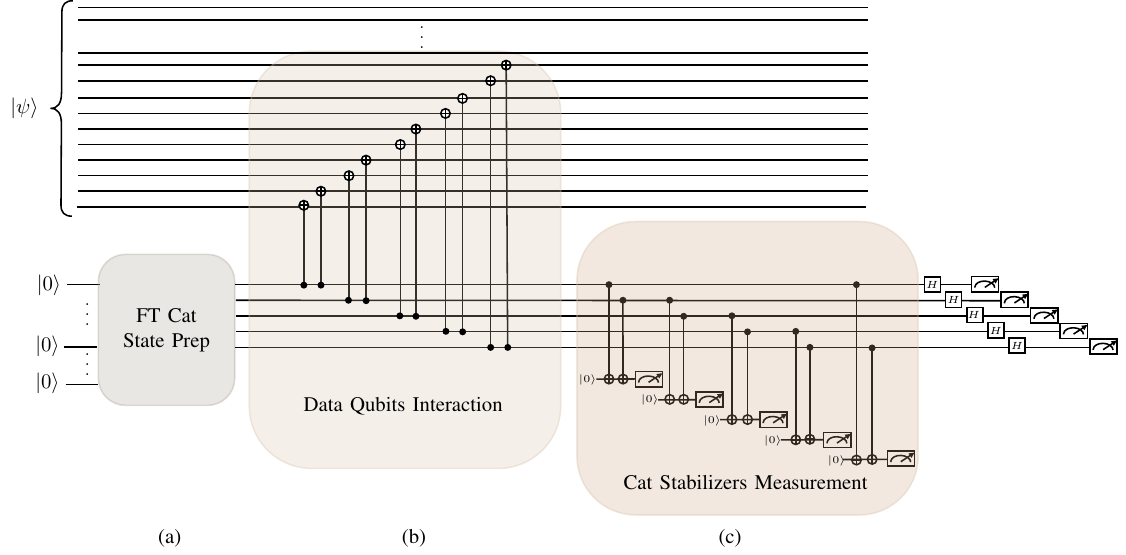} }
 	\caption{Cut-cat state $2$-\ac{FT} syndrome extraction protocol for an $\M{X}$-type stabilizer of weight $\gamma_i = 10$ in a code of distance $d = 5$.
    a) A cat state consisting of $\gamma_i / 2$ qubits is prepared.
    b) Each cat qubit acts as the control in two two-qubit gates that interact with data qubits.
    c) A single round of cat stabilizer measurements is performed.}
 	\label{Fig:CutCatState}
\end{figure*}

In this section, we describe the proposed cut-cat state $t$-\ac{FT} syndrome extraction protocol.
Fig.~\ref{Fig:CutCatState} illustrates an example of  syndrome extraction for an $\M{X}$-type stabilizer generator of weight $\gamma_i = 10$, in a code with distance $d = 5$.
Our proposed protocol can be summarized as follows.
First, a cat state of $\gamma_i / 2$ qubits is prepared fault-tolerantly, as described in Section~\ref{subsec:CatFTSE}.
Next, each cat qubit interacts with two different data qubits via two-qubit gates, arranged according to the structure of the generator being measured.
Finally, the protocol performs a variable number of rounds of $\M{Z}_i\M{Z}_j$ cat stabilizer measurements, where $i = 0, \dots, \gamma_i/2$ and $j = (i + 1) \bmod \left( \gamma_i / 2 \right)$, with the number of rounds determined by the code distance $d$.
We remark that it is possible to optimize the scheme by adaptively selecting a subset of cat stabilizers to measure after the first round, based on the outcomes obtained.
This adaptive approach can reduce the total number of measurements required, as we will show for code distance $d=7$.
The specific patterns of cat stabilizers to be measured are discussed in Section~\ref{sec:Decoding} and Section~\ref{sec:Proofs}.
Importantly, for the technique to be $t$-\ac{FT}, it is necessary to have at  least $d$ cat qubits.
The $d$ cat stabilizers form a ring of parity checks, and decoding proceeds by pairing the triggered checks and connecting them along the minimum path, as in a repetition code.
Having fewer than $d$ parity checks would not provide enough redundancy to guarantee that, for any pattern of up to $t$ errors, the corresponding syndromes can always be paired and connected consistently for correction.

\begin{figure}[t]
 	\centering
 	\resizebox{0.9\columnwidth}{!}{
 	   \includegraphics{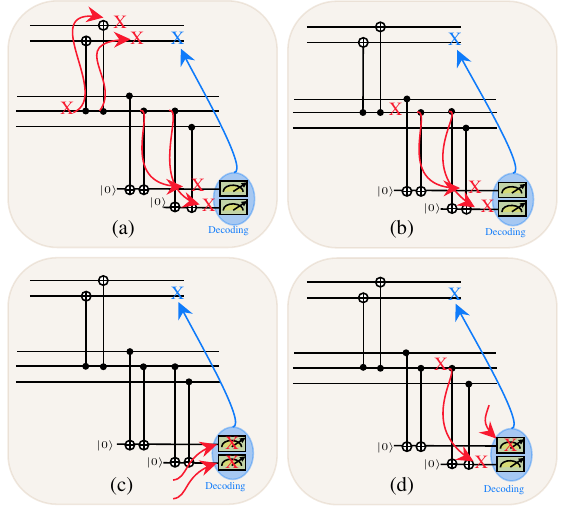}
     }
 	\caption{Error propagation in the cut-cat state syndrome extraction scheme.
    a) An $\M{X}$ error occurring before the two-qubit gates interacting with the data qubits propagates to a weight $w = 2$ data error.
    b) An $\M{X}$ error occurring after the final two-qubit gate interacting with the data qubits.
    c) Two cat stabilizer measurement errors.
    d) A combination of a cat stabilizer measurement error and an error induced by the first of the final two two-qubit gates used in the cat stabilizer measurement.
    All these error configurations produce the same cat stabilizers outcomes.
    The correction consists of applying a single $\M{X}$ operator to one of the two involved data qubits.}
 	\label{Fig:ErrorsCutState}
\end{figure}

Within this structure, a single $\M{X}$ error occurring on a cat qubit propagates to two data qubits, as shown in Fig.~\ref{Fig:ErrorsCutState}a.
Such an error must therefore be corrected in order to have a $t$-\ac{FT} scheme.
In the absence of additional errors, this fault activates two cat state stabilizers, enabling identification and correction of the affected data qubits.
However, an error occurring on the same cat qubit after both two-qubit gates interacting with the data qubits leads to the same cat stabilizer measurement outcome as in the cases of: (i) two measurement errors, or (ii) a measurement error combined with an error on the same cat qubit induced by the first of the final two two-qubit gates used in the cat stabilizer measurement, as illustrated in Fig.~\ref{Fig:ErrorsCutState}b, Fig.~\ref{Fig:ErrorsCutState}c, and Fig.~\ref{Fig:ErrorsCutState}d. 
Hence, the correction consists of applying a single $\M{X}$ operator to one of the two data qubits identified by the triggered cat stabilizers.
In all the exemplified scenarios, the result is a single data qubit error caused by either one or two faults, thus satisfying the fault-tolerance conditions. 
More generally, the decoding strategies used to correct errors based on the cat stabilizer outcomes are detailed in Section~\ref{sec:Decoding}.

\begin{table*}[t]
    \centering
    \caption{Comparison between FT Syndrome Extraction Strategies with Cat State Error Correction}
    \label{tab:ErrCorr}
    \small
    \includegraphics[width=\textwidth]{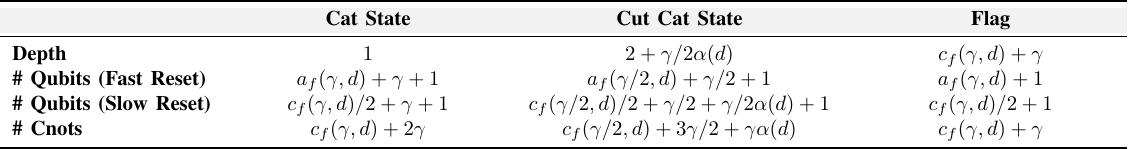}
\end{table*}

\begin{table*}[t]
    \centering
    \setlength{\tabcolsep}{3pt}
    \caption{Comparison between FT Syndrome Extraction Strategies with Cat State Error Detection}
    \label{tab:ErrDet}
    \small
    \includegraphics[width=\textwidth]{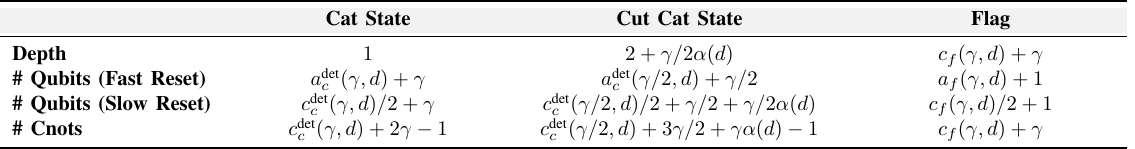}
\end{table*}

In Tab.~\ref{tab:ErrCorr} and Tab.~\ref{tab:ErrDet} we present a comparison between various $t$-\ac{FT} syndrome extraction protocols based on circuit depth, qubit count, and number of two-qubit gates.
We carry out two comparisons based on the cat state preparation. In Tab.~\ref{tab:ErrCorr} error correction-based preparation is assumed, while in Tab.~\ref{tab:ErrDet} the error detection-based is assumed.
In the error correction scenario, we assume that the cat state is prepared by fault-tolerantly measuring its $\M{X}$ stabilizer generator, which consists of a Pauli $\M{X}$ operator acting on each qubit of the cat state.
Specifically, $c_f(\rev{\gamma}, d)$ and $a_f(\rev{\gamma}, d)$ denote, respectively, the number of two-qubit gates interacting with flags and the number of simultaneous flag qubits required by a $t$-\ac{FT} syndrome extraction protocol using flag qubits for error correction to measure a stabilizer generator of weight $\rev{\gamma}$ in a code of distance $d$.
Similarly, $c_c^{\text{det}}(\rev{\gamma}, d)$ and $a_c^{\text{det}}(\rev{\gamma}, d)$ represent, respectively, the number of two-qubit gates interacting with flags and the number of simultaneous flag qubits required by a $t$-\ac{FT} cat state preparation based on quantum error detection, measuring a stabilizer generator of weight $\rev{\gamma}$.
Let us denote as $\alpha(d)$ the average number of cat stabilizer measurement per cat qubit.
For codes with distance $d = 3, 5, 7, 9$, we set $\alpha(d,\rev{\gamma}) = \lfloor (d+2)/4 \rfloor$ when a non-adaptive strategy is used for the cat stabilizer measurement.
As an example, under an adaptive strategy, the value of $\alpha(7,\rev{\gamma})$ can be reduced from $2$ to $1 + 8/\rev{\gamma}$ for $\rev{\gamma}>8$, as demonstrated in Section~\ref{sec:Proofs}. 
Moreover, we consider both the fast reset case, in which a qubit can be immediately reused after measurement, and the slow reset case, in which a qubit cannot be reset during the current syndrome extraction round.
From this comparison, it appears that the described $t$-\ac{FT} technique represents a trade-off between cat state-based and flag-based syndrome extraction schemes, while preserving their key advantages.
In the fast reset scenario, it requires less than half the number of simultaneous qubits compared to standard cat state syndrome extraction.
Moreover, if the physical implementation supports parallelization of two-qubit gates, the proposed scheme achieves a very low circuit depth compared with the respective flag-based syndrome extraction protocols.
Note that for cat-based approaches, since cat state preparation is entirely independent of the data qubits, we exclude the preparation circuit from the circuit depth analysis. 
However, we still account for both the required qubits and the number of two-qubit gates.
In Section~\ref{sec:NumRes}, we provide examples of the circuit depth, qubit count, and number of CNOT gates for various code distances and stabilizer generator weights.

In the case of an odd-weight operator, or an operator with weight $\rev{\gamma} < 2d$, the same technique can be applied by allowing one or more cat qubits to interact with a single data qubit.

\section{Decoding}
\label{sec:Decoding}
In this section, we introduce decoding algorithms tailored to the proposed $t$-\ac{FT} syndrome extraction scheme. 
For quantum codes with distance up to $d = 7$, we provide exact, rule-based decoding procedures. 
For codes of distance $d = 9$, we employ a \ac{LUT} approach to perform the decoding.
Since the proposed protocol is applicable to \ac{CSS} codes, we will assume throughout that we are measuring an $\M{X}$ Pauli operator.  
The decoding scheme for an $\M{Z}$ Pauli operator can be obtained analogously by interchanging the roles of the $\M{X}$ and $\M{Z}$ operators.
In the following, we refer to the cut-cat state syndrome extraction gadget in Sec.~\ref{sec:Cat?FT}. 
Without loss of generality, we assume that when measuring a generator of weight $\gamma_i/2$, the $r$-th cat qubit interacts with the data qubits $q_{2r}$ and $q_{2r + 1}$ for $r = 0, \dots, \gamma_i/2 - 1$.
Moreover, the set of cat stabilizers that measure the $r$-th cat qubit during the $i$-th round of cat stabilizer measurements is defined as $\mathcal{G}_{r,i}^{\text{cat}} = \left\{M_{r-1}^{(i)}, M_{r}^{(i)} \right\}$, where $M_{r}^{(i)}$ is the measurement outcome of the $r$-th cat stabilizer during round $i$-th.
Similarly, we refer to the set containing every cat stabilizer measurements as $\mathcal{G}^{cat}_{i}$.
Also, all subscript indices referred to cat stabilizers are to be interpreted modulo $\gamma_i/2$. 
For instance, the measurement outcome $M_{r+1}^{(i)}$, should be always understood as $M_{ (r + 1) \bmod \left( \gamma_i/2 \right)}^{(i)}$.

We define the set of ordered switched on cat stabilizers in the first round as $\mathcal{A} = \left\{a_0, \dots, a_{n_{\mathrm{a}-1}}\right\}$, where $n_{\mathrm{a}}$ is the total number of switched on cat stabilizers. 
Recall that the first round can be used to adaptively choose the measurements for the next round.
Given a generator of weight $\gamma_i$, we define a subset of cat qubit indices as
$\mathcal{I}_{m} = \left\{ (a_m + k) \bmod \left( \gamma_i/2 \right) : 0 \leq k \leq a_{m+1} - 1 - a_m \right\}.$
Moreover, given a collection of such sets, we define their union as
$\mathcal{I} = \bigcup_{m \text{ even}} \,\, \mathcal{I}_{m}$.
The set of complementary indices is defined by
$\mathcal{I}^c = \mathcal{I}_\Omega \setminus \mathcal{I}$, where $\mathcal{I}_\Omega$ denotes the set of all cat qubit indices.

The decoding algorithm adopted for codes of distance $d = 3$ and $d = 5$ is described in Algorithm~\ref{algo:D3/5}.
The decoding algorithm adopted for codes of distance $d = 7$ is described in Algorithms~\ref{algo:D7_1} to~\ref{algo:D7_6}.
For the sake of clarity, we subdivide this algorithm according to $|\mathcal{A}|$.
In writing the algorithm, we make use of lists of subsets of cat stabilizer indices, which we denote by $\mathcal{J}_{n_\mathrm{a}, k}$.
These lists are explicitly provided in Appendix~\ref{app:Lists7}.
The notation $M^{(i)}\left[\mathcal{J}\right]$ denotes the list of measurement outcomes associated with the cat stabilizer indices $\mathcal{J}$.
Therefore, equalities and inequalities on two lists $M^{(i)}\left[\mathcal{J}\right]$ should be interpreted elementwise; equalities hold only if all corresponding elements are equal, while inequalities hold if at least one corresponding pair is unequal.
We also define the function \texttt{inv()} which, given a binary string, returns its bitwise complement.

The decoder requires a single pass over a binary array of size $\gamma_i/2$, yielding $O(\gamma_i)$ time complexity.

\begin{algorithm}[t]
\SetKwInOut{Input}{input}
\SetKwInOut{Output}{output}
\caption{Distance Three/Five Code, $|\mathcal{A}|$ even} 
\label{algo:D3/5}
compute $\mathcal{I}$, $\mathcal{I}^c$;\\
\If{$|\mathcal{I}| < |\mathcal{I}^c|$}{
\If{$\left(|\mathcal{I}| \leq t \text{ and } |\mathcal{A}| \leq 2t\right)$ \text{or} $|\mathcal{A}| > 2t$}{
apply $\M{X}$ on $q_{2\mathcal{I} + 1}$;
}
} \Else {
\If{$\left(|\mathcal{I}^c| \leq t \text{ and } |\mathcal{A}| \leq 2t\right)$ \text{or} $|\mathcal{A}| > 2t$}{
apply $\M{X}$ on $q_{2\mathcal{I}^c + 1}$;
}
}
\end{algorithm} 

\begin{algorithm}[t]
\SetKwInOut{Input}{input}
\SetKwInOut{Output}{output}
\caption{Distance Seven Code, $|\mathcal{A}| = 1$} \label{algo:D7_1}
\uIf{$M^{(1)}\left[\mathcal{J}_{1, 0}(a_0)\right] \neq M^{(0)}\left[\mathcal{J}_{1, 0}(a_0)\right]$}
{apply $\M{X}$ on $q_{2 a_0 + 3}$;} 
\uElseIf{$M^{(1)}\left[\mathcal{J}_{1, 1}(a_0)\right] \neq M^{(0)}\left[\mathcal{J}_{1, 1}(a_0)\right]$}
{apply $\M{X}$ on $q_{2 a_0 + 1}$;
}
\end{algorithm} 

\begin{algorithm}[t]
\SetKwInOut{Input}{input}
\SetKwInOut{Output}{output}
\caption{Distance Seven Code, $|\mathcal{A}| = 2$} 
$\mathrm{flagDecoding \gets 1}$; \\
\uIf{$a_0 + 3 \in \mathcal{A}$}{
\uIf{$M^{(1)}\left[\mathcal{J}_{2, 0}(a_0)\right] \neq M^{(0)}\left[\mathcal{J}_{2, 0}(a_0)\right]$}
{$\mathrm{flagDecoding \gets 0}$; } 
}\uElseIf{$a_0 - 3 \in \mathcal{A}$}{
\uIf{$M^{(1)}\left[\mathcal{J}_{2, 1}(a_0)\right] \neq M^{(0)}\left[\mathcal{J}_{2, 1}(a_0)\right]$}
{$\mathrm{flagDecoding \gets 0}$; } 
}\uElseIf{$a_0 + 2 \in \mathcal{A}$}{
\uIf{$M_{a_0}^{(1)} \neq M_{a_0}^{(0)}$} {
\uIf{$M^{(1)}\left[\mathcal{J}_{2, 2}(a_0)\right] \neq M^{(0)}\left[\mathcal{J}_{2, 2}(a_0)\right]$}
{$\mathrm{flagDecoding \gets 0}$; } 
} \uElseIf{$M_{a_0 +2}^{(1)} \neq M_{a_0 + 2}^{(0)}$} {
\uIf{$M^{(1)}\left[\mathcal{J}_{2, 3}(a_0)\right] \neq M^{(0)}\left[\mathcal{J}_{2, 3}(a_0)\right]$}
{$\mathrm{flagDecoding \gets 0}$; } 
} \uElseIf{$M^{(1)}\left[\mathcal{J}_{2, 4}(a_0)\right] \neq M^{(0)}\left[\mathcal{J}_{2, 4}(a_0)\right]$}
{apply $\M{X}$ on $q_{2 a_0 + 3}$; \\
$\mathrm{flagDecoding \gets 0}$; }
}
\uIf{$\mathrm{flagDecoding = 1}$}{
compute $\mathcal{I}$, $\mathcal{I}^c$;\\
\uIf{$|\mathcal{I}| \leq 3$}{
apply $\M{X}$ on $q_{2\mathcal{I} + 1}$;
} \uElseIf{$|\mathcal{I}^c| \leq 3$} {
apply $\M{X}$ on $q_{2\mathcal{I}^c + 1}$;
}
}
\end{algorithm} 

\begin{algorithm}[t]
\SetKwInOut{Input}{input}
\SetKwInOut{Output}{output}
\caption{Distance Seven Code, $|\mathcal{A}| = 3$}  \label{algo:D7_3}
$\mathrm{flagDecoding \gets 1}$; \\
\uIf{$\{a_0 + 2, a_0 + 4\} \in \mathcal{A}$}{
\uIf{$\texttt{inv}\left(M^{(1)}\left[\mathcal{J}_{3, 0}(a_0)\right]\right) =  M^{(0)}\left[\mathcal{J}_{3, 0}(a_0)\right]$} {apply $\M{X}$ on $q_{2 (a_0 + 3) + 3}$;
} \uElse{
{apply $\M{X}$ on $q_{2 a_0 + 3}$; 
} 
}
$\mathrm{flagDecoding \gets 0}$;
}
\uElseIf{$\{a_0 + 1, a_0 + 2\} \in \mathcal{A}$}
{
\uIf{$M_{a_0 +1}^{(1)} \neq M_{a_0 +1}^{(0)}$}{
apply $\M{X}$ on $q_{2 (a_0 + 1) + 1}$;
}
\uElseIf{$M_{a_0 -1}^{(1)} \neq M_{a_0 -1}^{(0)}$}{
apply $\M{X}$ on $q_{2 (a_0 + 1) + 3}$;
}
\uElseIf{$M_{a_0}^{(1)} = M_{a_0}^{(0)}$}{
\uIf{$M_{a_0 - 2}^{(1)} \neq M_{a_0 - 2}^{(0)}$}{
apply $\M{X}$ on $q_{2 (a_0 + 1) + 3}$;
} 
\uElse{ 
apply $\M{X}$ on $q_{2 (a_0 + 1) + 1}$;
}
}
\uElse{
apply $\M{X}$ on $q_{2 (a_0 + 1) + 3}$;
}
$\mathrm{flagDecoding \gets 0}$;
}
\uElseIf{$\{a_0 + 2, a_0 + 3\} \in \mathcal{A}$}
{
\uIf{$M^{(1)}\left[\mathcal{J}_{3, 1}(a_0)\right] \neq M^{(0)}\left[\mathcal{J}_{3, 1}(a_0)\right]$}
{
apply $\M{X}$ on $q_{2 (a_0 + 1) + 1}$; \\
$\mathrm{flagDecoding \gets 0}$;
}
}
\uElseIf{$\{a_0 + 1, a_0 + 3\} \in \mathcal{A}$}
{
\uIf{$M^{(1)}\left[\mathcal{J}_{3, 2}(a_0)\right] \neq M^{(0)}\left[\mathcal{J}_{3, 2}(a_0)\right]$}
{
apply $\M{X}$ on $q_{2 (a_0 + 1) + 3}$; \\
$\mathrm{flagDecoding \gets 0}$;
}
}
\uIf{$\mathrm{flagDecoding = 1}$}
{
\ForAll{$g \in \mathcal{G}^{cat}_0$}
{
\uIf{$\{g, g+1\} \in \mathcal{A}$}
{
apply $\M{X}$ on $q_{2 g + 3}$; \\
$\mathrm{flagDecoding \gets 0}$;
}
}
}
\uIf{$\mathrm{flagDecoding = 1}$}
{
\ForAll{$g \in \mathcal{G}^{cat}_0$}
{
\uIf{$\{g, g+2\} \in \mathcal{A}$}
{
apply $\M{X}$ on $q_{2 g + 3}$; 
}
}
}
\end{algorithm} 

\begin{algorithm}[t]
\SetKwInOut{Input}{input}
\SetKwInOut{Output}{output}
\caption{Distance Seven Code, $|\mathcal{A}| = 4$} \label{algo:D7_4}
$\mathrm{flagDecoding \gets 1}$; \\
\uIf{$a_0 + 2 \in \mathcal{A}$}{
\uIf{$a_0 + 1 \in \mathcal{A}$}{
\uIf{$a_0 + 3 \in \mathcal{A}$}{
\uIf{$M^{(1)}\left[\mathcal{J}_{4, 0}(a_0)\right] \neq M^{(0)}\left[\mathcal{J}_{4, 0}(a_0)\right]$} {
\uIf{$M^{(1)}\left[\mathcal{J}_{4, 1}(a_0)\right] \neq M^{(0)}\left[\mathcal{J}_{4, 1}(a_0)\right]$} {
$\mathrm{flagDecoding \gets 0}$;
}
}
}
\uElseIf{$a_0 + 4 \in \mathcal{A}$} {
\uIf{$M^{(1)}\left[\mathcal{J}_{4, 2}(a_0)\right] \neq M^{(0)}\left[\mathcal{J}_{4, 2}(a_0)\right]$} 
{
$\mathrm{flagDecoding \gets 0}$;
}
}
}
\uElseIf{$\{a_0 +3, a_0 + 4\} \in \mathcal{A}$}
{
\uIf{$M^{(1)}\left[\mathcal{J}_{4, 3}(a_0)\right] \neq M^{(0)}\left[\mathcal{J}_{4, 3}(a_0)\right]$}{
$\mathrm{flagDecoding \gets 0}$;
}
}
}
\uIf{$\mathrm{flagDecoding = 1}$}{
compute $\mathcal{I}$, $\mathcal{I}^c$;\\
\uIf{$|\mathcal{I}| \leq 3$}{
apply $\M{X}$ on $q_{2\mathcal{I} + 1}$;
} \uElseIf{$|\mathcal{I}^c| \leq 3$} {
apply $\M{X}$ on $q_{2\mathcal{I}^c + 1}$;
}
}

\end{algorithm} 

\begin{algorithm}[t]
\SetKwInOut{Input}{input}
\SetKwInOut{Output}{output}
\caption{Distance Seven Code, $|\mathcal{A}| = 5$} \label{algo:D7_5}
$\mathrm{flagDecoding \gets 1}$; \\
\uIf{$\{a_0 +1, a_0 + 2,a_0 +3, a_0 + 4\} \in \mathcal{A}$}{
\uIf{$M^{(1)}\left[\mathcal{J}_{5, 0}(a_0)\right] = M^{(0)}\left[\mathcal{J}_{5, 0}(a_0)\right]$}
{
apply $\M{X}$ on $q_{2 (a_0 + 3) + 3}$;
}
\uElse{
apply $\M{X}$ on $q_{2 (a_0 + 2) + 3}$;
}
$\mathrm{flagDecoding \gets 0}$; 
}
\uIf{
$\mathrm{flagDecoding = 1}$
}{
\ForAll{$g \in \mathcal{G}^{cat}_0$}
{
\uIf{$\{g, g+1\} \in \mathcal{A}$}
{
\uIf{$\{g -1, g+2\} \notin \mathcal{A}$ \text{or} $\{g +2, g+3\} \in \mathcal{A}$}{
apply $\M{X}$ on $q_{2 (g+1) + 1}$;
}
}
}
}
\end{algorithm} 

\begin{algorithm}[t]
\SetKwInOut{Input}{input}
\SetKwInOut{Output}{output}
\caption{Distance Seven Code, $|\mathcal{A}| \geq 6$, $|\mathcal{A}|$ even} \label{algo:D7_6}
compute $\mathcal{I}$, $\mathcal{I}^c$;\\
\If{$|\mathcal{I}| < |\mathcal{I}^c|$}{
\If{$\left(|\mathcal{I}| \leq 3 \text{ and } |\mathcal{A}| = 6\right)$ \text{or} $|\mathcal{A}| > 6$}{
apply $\M{X}$ on $q_{2\mathcal{I} + 1}$;
}
} \Else {
\If{$\left(|\mathcal{I}^c| \leq 3 \text{ and } |\mathcal{A}| = 6\right)$ \text{or} $|\mathcal{A}| > 6$}{
apply $\M{X}$ on $q_{2\mathcal{I}^c + 1}$;
}
}
\end{algorithm}

\section{Fault-Tolerant Syndrome Extraction Gadget}
\label{sec:Proofs}
In this section, we prove the fault-tolerance of the proposed syndrome extraction scheme when decoded by using algorithms described in Sec.~\ref{sec:Decoding}. 
In doing so, we consider only errors of weight $w \leq t$, as these are the ones that must be controlled to ensure fault-tolerance. 

We first show that errors of weight $w \leq t$ never propagate to higher-weight data errors during syndrome extraction using the cut-cat scheme.
We then prove that repeating the syndrome extraction measurement until $t+1$ identical syndromes are obtained consecutively yields a \ac{FT} syndrome extraction gadget.
Hence, in Theorems~\ref{th:FTdist3}–\ref{th:FTdist7} we will focus on $(i)$ errors occurring before the start of the circuit on each cat qubit,
$(ii)$ errors on the control qubit of each two-qubit gate, occurring after the gate's execution, and
$(iii)$ errors occurring just before each cut-cat measurement.
These fault locations include all errors that can propagate to higher-weight data errors or modify the cut-cat measured syndrome during the execution of the circuit.
Indeed, note that errors on data qubits present prior to syndrome extraction, or errors generated on the target of a two-qubit gate, never propagate.
Finally, in Theorem~\ref{th:FTSE} we will also account for faults affecting the measurement outcomes of the syndrome qubits corresponding to the generator.

\subsection{Distance Three}

\begin{theorem}
\label{th:FTdist3}

For a quantum code with $d = 3$, the technique proposed in Sec.~\ref{sec:Cat?FT}, using a single round of cat stabilizer measurements, ensures that any single fault never propagates to a data error of weight greater than one.

\end{theorem}

\begin{proof}
When measuring a code generator, a single error may cause one or two cat stabilizer measurements to trigger.
Specifically, an error occurring before the interaction with the data qubits can propagate into a weight $w = 2$ data error, which anticommutes with a pair of cat state stabilizer generators.
However, the same syndrome pattern may also be caused by an error on the same cat qubit occurring after its first two-qubit gate interaction with a data qubit, or even after both interactions.
As a result, when a pair of anticommuting cat state generators is triggered and both act on the $r$-th cat qubit, we conservatively correct an error on the data qubit $q_{2r + 1}$.
In this way, if the error occurs before the interaction with the data qubits, one of the two propagated errors is corrected; if it occurs between the two-qubit gates, the single propagated error is corrected; and if it occurs after the two-qubit gates, only a single error is introduced. In all three cases, a single fault results in at most one error on the data qubit.
As a final possibility, an error may occur during a cat stabilizer measurement. In this case, a single cat syndrome measurement is triggered. 
To avoid introducing a data error, we simply ignore it and do not apply any correction. 
All other possible cat syndrome patterns must arise from errors of weight $w>1$ during the execution of the syndrome extraction gadgets. 
Consequently, they do not affect the fault-tolerance analysis for $d=3$.
\end{proof}

\subsection{Distance Five}

\begin{theorem}
\label{th:FTdist5}
For a quantum code with $d = 5$, the technique proposed in Sec.~\ref{sec:Cat?FT}, using a single round of cat stabilizer measurements, ensures that any fault of weight $w \leq~2$ never propagates to a data error of weight greater than $w$.
\end{theorem}

\begin{proof}
When measuring a code generator, an error of weight $w \leq 2$ may cause up to four cat stabilizer measurements to trigger.
To proceed with the derivation, we divide the analysis into cases based on the number of triggered cat stabilizers.

If no cat stabilizers are triggered, then: $i)$ either no error has occurred; $ii)$ at least two errors have occurred on the same cat qubit prior to measurement; $iii)$ two errors have canceled each other out without interacting with the data qubits.
For both cases $i)$ and $iii)$, no correction is required.
In the worst-case scenario highlighted by $ii)$, an error occurring before the interaction with the data qubits may propagate to a data error of weight $w = 2$, which is then canceled by a second error occurring after both 2-qubit gate interactions with data qubits. 
In this case, no correction is applied, since an error of weight two is propagated to a data error of weight $w = 2$.
Hence, when no cat stabilizers are triggered, no correction is applied.

If a single cat stabilizer is triggered, the cause is either a single measurement error, or a combination of an error on a cat qubit and a measurement error.
In the worst-case scenario, an error occurring before the interaction with the data qubits may propagate to a data error of weight $w = 2$. This would normally trigger two adjacent cat stabilizer measurements; however, one of them fails to register due to a measurement error.
In this case, no correction is applied, since an error of weight two is propagated to a data error of weight $w = 2$.

If two cat stabilizers are triggered, this may result from either a single or double cat qubit error, or from two measurement errors.
In particular, if the triggered measurement outcomes are $M_{r-1}^{(i)}$ and $M_{r}^{(i)}$ (i.e., are adjacent), we apply a correction on the data qubit $q_{2r+1}$.
On the other hand, if the triggered stabilizers are $M_{r-1}^{(i)}$ and $M_{r+1}^{(i)}$, we apply corrections on the data qubits $q_{2r+1}$ and $q_{2(r+1) +1}$.
In this way, if two cat qubit errors occur before the interaction with the data qubits, they propagate to four data qubits. 
Two of these will be corrected, leaving a residual data error of weight $w = 2$.
Recall that the number of cat qubits must be at least equal to the code distance, so that any error of weight $w \leq 2$ can always be correctly identified by the cat stabilizer measurements.
In all other cases, no correction is applied. This situation arises, for example, when two measurement errors occur.

If three cat stabilizers are triggered, this must be due to a cat qubit error combined with a measurement error. 
In this case, an error of weight $w = 2$ may propagate to a data error of weight $w \leq 2$, and therefore no correction is applied.

Finally, if four cat stabilizers are triggered, this must be due to two non-adjacent cat qubit errors. 
Specifically, errors occurring on the $r$-th and $u$-th cat qubits will trigger the stabilizers $M_{r-1}^{(i)}$, $M_{r}^{(i)}$, $M_{u-1}^{(i)}$, and $M_{u}^{(i)}$. 
In this case, a correction must be applied to the data qubits $q_{2r+1}$ and $q_{2u+1}$.

\end{proof}

\subsection{Distance Seven}

\begin{theorem}
\label{th:FTdist7}
For a quantum code with $d = 7$, the technique proposed in Sec.~\ref{sec:Cat?FT}, using a single round of cat stabilizer measurements, along with up to four additional cat stabilizer measurements chosen adaptively based on the outcomes of the first round, ensures that any fault of weight $w \leq~3$ never propagates to a data error of weight greater than $w$.

\end{theorem}

\begin{proof}
A detailed proof is provided in Appendix~\ref{app:FTdist7} to maintain the clarity of the main exposition.

\end{proof}

\subsection{Distance Nine and Beyond}

For codes of distance $d = 9$ and various values of $\gamma_i$, the proposed scheme includes two complete rounds of cat stabilizer measurements, and has been decoded using a \ac{LUT} approach.
This \ac{LUT} was constructed via exhaustive search by inserting all combinations of faults of weight up to $t$, considering the following types of errors:
$(i)$ errors occurring before the start of the circuit on each cat qubit,
$(ii)$ errors on the control qubit of each two-qubit gate, occurring after the gate's execution, and
$(iii)$ errors occurring just before each cut-cat measurement.
These fault locations include all errors that can propagate to higher-weight data errors or modify the measured cut-cat syndrome during the execution of the circuit, while faults at other locations do not propagate to higher-weight errors and therefore do not require inclusion in the \ac{LUT}.
Specifically, the \ac{LUT} assigns a fixed correction to each corresponding cat measurement outcome, ensuring that no error pattern of weight up to $t$, consistent with that outcome, propagates to a data error of weight greater than $t$ after correction (up to a stabilizer).
In this way, we ensure that all errors within the correction capability of the code never propagate to higher-weight errors.
A pseudocode description of the \ac{LUT} construction for codes of distance $d = 9$ can be found in Appendix~\ref{app:LUT}.

We conjecture that the proposed scheme, when supplemented with additional rounds of cat stabilizer measurements, can serve as a $t$-\ac{FT} syndrome extraction gadget for CSS codes of arbitrary distance.

\subsection{Proof of Fault-Tolerant Syndrome Extraction}
In the following, we show that the cut-cat state syndrome extraction gadget, when combined with a standard \ac{FT} measurement protocol based on repeated syndrome extraction, constitutes a \ac{FT} syndrome extraction gadget.

\begin{theorem}[]
\label{th:FTSE}
Let us consider a quantum code of distance $d$. 
When the cut-cat state syndrome extraction gadget is combined with a measurement protocol in which syndrome extraction is repeated until $t+1$ identical syndromes are obtained consecutively~\cite{Sho96:CatState, TanBal23:t+1RepeatedSyndromeExraction}, the resulting procedure constitutes a \ac{FT} syndrome extraction gadget according to Definition~\ref{def:FTEC}.
\end{theorem}

\begin{proof}
Theorems~\ref{th:FTdist3}–\ref{th:FTdist7} and the \acp{LUT} construction ensure that any set of $s \le t$ faults whose errors propagate from syndrome to data qubits during a round of syndrome extraction produces a data error of weight at most $s$.
Moreover, measurement faults are handled by the standard Shor-style repetition protocol in which syndrome extraction is repeated until $t+1$ identical syndromes are obtained consecutively~\cite{Sho96:CatState}. 
In particular, if at most $t$ measurement errors occur during the repeated extraction procedure, then an incorrect syndrome is not accepted.
Combining these two facts, an incoming data error of weight $r$ together with $s$ additional faults during syndrome extraction, with $r+s \le t$, results in a total data error of weight at most $t$. Such an error is always correctable by an ideal decoder, leading to a logical error with probability $\mathcal{O}(p^{t+1})$.
\end{proof}

\section{Numerical Results}\label{sec:NumRes}

This section presents a numerical analysis of the fault-tolerance of the proposed syndrome extraction scheme using Monte Carlo simulations.
All numerical simulations are performed by running the decoder until a minimum of 100 errors are reached, ensuring reliable results.

\emph{1) Gadget Validation:} Let $\rv{S}$ denote the random variable representing the number of errors introduced by the syndrome extraction gadget that are propagated into the data qubits.
Also, let $p_\mathrm{FT}$ denote the physical error rate per gate and measurements.
In our numerical analysis, we simulate the syndrome extraction process varying the code distance $d\in \{3, 5, 7, 9\}$. 
To model physical noise, we apply an initial depolarizing channel to each cat qubit. Subsequently, a single-qubit depolarizing channel is applied after every single-qubit gate, and a two-qubit depolarizing channel is applied after every two-qubit gate. All resulting errors are then propagated through the entire circuit. Finally, a depolarizing channel is applied immediately before each measurement. The depolarizing probability for all channels is consistently denoted by $p_\mathrm{FT}$.
Fig.~\ref{Fig:BDLogErr} shows the probability that an error of weight $w > t$ is propagated to the data qubits by the cut-cat state syndrome extraction gadget as a function of the physical gate error probability, for $\gamma_i = 2d$.
In particular, the obtained results match the expected slope $\mathcal{O}(p^{t+1})$ for a code of distance $d = 2t + 1$, depicted in dashed lines.
Together with the theorems, this provides further evidence that our scheme is $t$-\ac{FT}, as it prevents low-weight errors ($w < t$) from propagating into higher-weight faults.
Fig.~\ref{Fig:HalfCatt_Reduced} shows an analogous plot for the case $\gamma_i < 2d$.
In this regime, $d$ cat qubits are still required to correctly decode the cat-state syndrome, but one or more of these cat qubits interacts with only a single data qubit.
For example, for $d = 5$ and $\gamma_i = 8$, three cat qubits interact with two data qubits, while the remaining two cat qubits interact with one data qubit.
The decoding algorithm remains the same as that described in Section~\ref{sec:Decoding}; however, when applying the final correction, the indices of some data qubits must be modified to take into account that some cat qubits interact with only a single data qubit.
In this case as well, the numerical results follow the expected scaling $\mathcal{O}(p^{t+1})$ for a code of distance $d = 2t + 1$, as indicated by the dashed lines.

We provide an upper bound on the probability that a fault pattern propagates to a data error of weight $w > t$, for arbitrary code distance $d$, given the weight $\gamma_i$ of the stabilizer generator measured using a cut-cat state syndrome-extraction circuit.
We identify all fault locations in the syndrome extraction gadget from which a fault can propagate to the data qubits or cause a cut-cat syndrome measurement error. 
Let $N_{\mathrm{FL}}$ denote the total number of such fault locations.
For the cut-cat state syndrome extraction circuit, assuming that each cat qubit interacts with two data qubits, the number of relevant fault locations scales linearly with the stabilizer weight. 
Specifically, each cat qubit participates in four CNOT gates whose control errors can propagate either to data qubits or to the cut-cat syndrome. 
Hence, each cat qubit contributes four fault locations. 
Likewise, each cut-cat syndrome qubit contributes four additional fault locations: one state preparation, two CNOT targets, and one measurement. Since a stabilizer of weight $\gamma_i$ is measured using $\gamma_i/2$ cat qubits, the total number of relevant fault locations is
\begin{align}
N_{\mathrm{FL}} = \frac{\gamma_i}{2} (4 + 4) = 4 \gamma_i .
\end{align}

Under a circuit-level noise model with physical gate error probability $p_{\mathrm{FT}}$, an upper bound on the probability that more than $t$ faults occur at these locations, thereby potentially producing a data error of weight $w > t$, is given by
\begin{align}
\label{eq:UBcat}
\Prob{\rv{S} > t}
\leq 1 - \sum_{s=0}^{t}
\binom{N_{\mathrm{FL}}}{s}
\left(\frac{2 p_{\mathrm{FT}}}{3}\right)^{s}
\left(1 - \frac{2 p_{\mathrm{FT}}}{3}\right)^{N_{\mathrm{FL}} - s}
\end{align}
where the factor of $2/3$ arises because, for a depolarizing error on a faulty gate, only two of the three Pauli errors either propagate to the data qubits or induce a cut-cat syndrome measurement error.
This upper bound is shown in Fig.~\ref{Fig:BDLogErr} and is consistent with the observed numerical behavior.

\begin{figure}[t]
 	\centering
 	\resizebox{0.99\columnwidth}{!}{
 	   \includegraphics{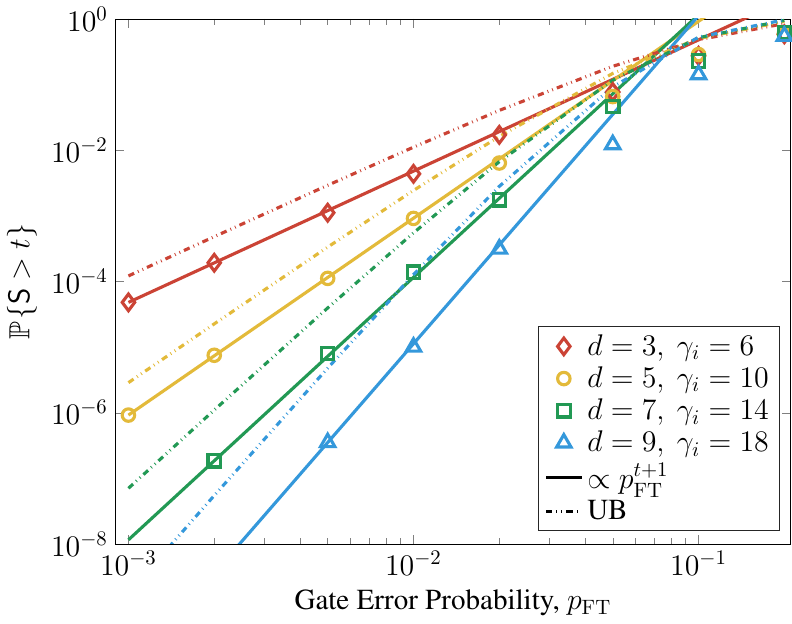}
     }
 	\caption{Probability that an error of weight $w > t$ is propagated to the data qubits by the cut-cat state syndrome extraction gadget vs. gate error probability. Dashed lines indicate visually the expected $\mathcal{O}(p^{t+1})$ scaling for a code of distance $d=2t+1$, fitting the simulated data points.}
 	\label{Fig:BDLogErr}
\end{figure}

\emph{2) Code Block Analysis:}
To perform code block analysis, we employ the cut-cat state syndrome extraction gadget for the $[[49,1,5]]$ triorthogonal quantum code~\cite{TransversalClifford25:ShuVic}, whose triorthogonal structure enables transversal non-Clifford logical gates and therefore makes it particularly well suited for \ac{FT} magic-state distillation \cite{BraHaa:12, DagBluBro:25}.
This code admits a set of stabilizer generators with weights $\gamma_i = 4, 6, 8, 10$, one generator of weight $\gamma_i = 16$, and one generator of weight $\gamma_i = 32$.
We adopt the cut-cat state gadget to measure stabilizer generators with $\gamma_i > 6$, while full cat state syndrome extraction is used for the remaining generators.
All cat states are prepared fault-tolerantly using post-selection, following the procedure described in~\cite{ForAma25:FTCSS}.
For the noise model, a depolarizing channel with physical error rate $p$ is applied to each data qubit.
During the syndrome extraction procedure, a single-qubit depolarizing channel is applied after each qubit initialization and after every single-qubit gate, while a two-qubit depolarizing channel is applied after every two-qubit gate.
In addition, a depolarizing channel is applied immediately before each measurement.
The depolarizing probability associated with syndrome extraction is consistently denoted by $p_\mathrm{FT}$ and is chosen as $p/20$, $p/50$, or $p/100$, reflecting the requirement that error correction operations be more reliable than the physical noise affecting the data qubits.
The syndrome extraction procedure is repeated until $t+1$ identical syndromes are obtained~\cite{Sho96:CatState, TanBal23:t+1RepeatedSyndromeExraction}.
Decoding is then performed using a look-up table constructed by enumerating all error patterns up to the designed code distance, starting from minimum-weight errors and assigning to each previously unseen syndrome the corresponding correction.
As show in Fig.~\ref{Fig:HalfCat_[[49,1,5]]}, The resulting logical error rates exhibit the expected scaling $\mathcal{O}(p^{3})$ for a code of distance $d = 5$, shown by the dashed lines, thereby confirming the fault-tolerance of the proposed syndrome extraction gadget.

\begin{figure}[t]
 	\centering
 	\resizebox{0.99\columnwidth}{!}{
 	   \includegraphics{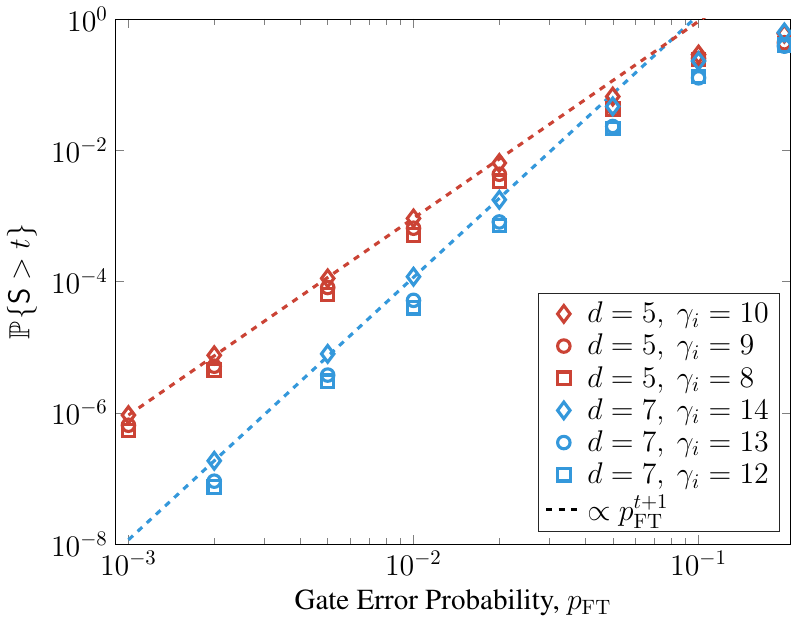}
     }
 	\caption{Probability that an error of weight $w > t$ is propagated to the data qubits by the cut-cat state syndrome extraction gadget vs. gate error probability. Dashed lines indicate visually the expected $\mathcal{O}(p^{t+1})$ scaling for a code of distance $d=2t+1$, fitting the simulated data points.}
 	\label{Fig:HalfCatt_Reduced}
\end{figure}
\begin{figure}[t]
 	\centering
 	\resizebox{0.99\columnwidth}{!}{
 	   \includegraphics{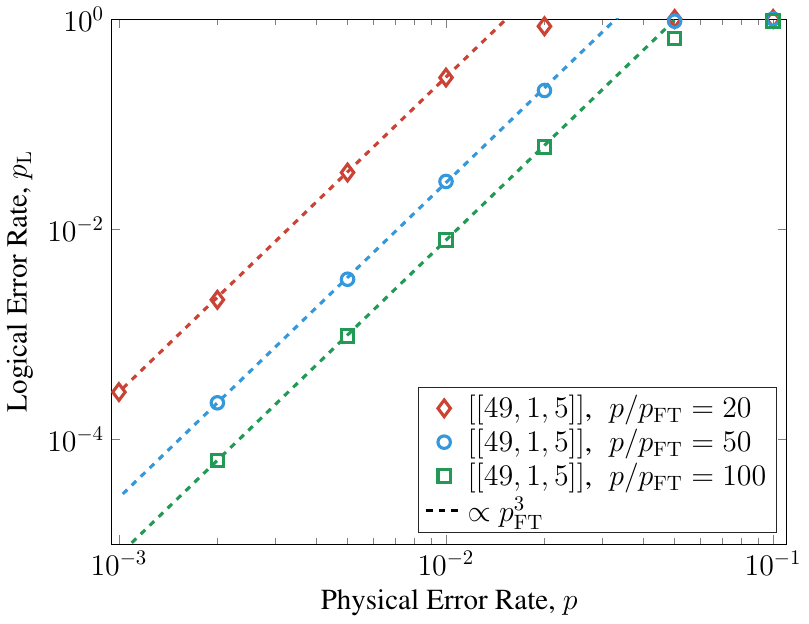}
     }
 	\caption{Logical error rate vs. physical error rate over a depolarizing channel. $[[49,1,5]]$ triorthogonal code employing the cut-cat state syndrome extraction gadget for each generator of weight $\gamma_i > 6$.
    Discontinuous lines indicate visually the expected $\mathcal{O}(p^{3})$ scaling for a code of distance $d=5$, fitting the simulated data points.}
 	\label{Fig:HalfCat_[[49,1,5]]}
\end{figure}

\emph{3) Resource Comparison with State-of-the-Art Methods:} In the following, we compare the proposed syndrome extraction scheme with the state-of-the-art alternatives from the literature, focusing on three key metrics: the required number of simultaneous qubits (assuming fast reset is available), the number of two-qubit gates, and the circuit depth.

In this comparison, we assume that two-qubit gates acting on disjoint pairs of qubits can be executed in parallel.
We further assume the availability of mid-circuit measurement and reset, which enables ancilla reuse across successive stabilizer measurements. If parallelism is limited, circuit depth increases.
In the absence of fast reset, the primary effect is an increased ancilla requirement.
In \cite{Pra23:OtimizedSchemeFlag5and7} \ac{FT} syndrome extraction schemes based on flag qubits for code distances $d = 5$ and $d = 7$ are introduced. For the measurement of a stabilizer operator with weight $\gamma_i = 14$ in a distance $d=7$ code, their approach requires 44 two-qubit gates and uses 8 physical qubits. The circuit depth is also 44, as each two-qubit gate shares the same control qubit and must be executed sequentially.
In the comparison with cat-based schemes, we assume a $t$-\ac{FT} preparation of the cat state using the postselection method described in \cite{ForAma25:FTCSS}. 
The number of two-qubit gates required by the cut-cat state approach to measure the same operator ranges from 40 to 48, depending on the observed cat-state syndrome. 
The number of simultaneous qubits required is 10.
Since cat state preparation does not involve data qubits, it can be performed offline.
Under this assumption, and assuming serial measurement of all cat stabilizers using a single ancilla qubit, the circuit depth ranges from 9 to 13.
For comparison, the full cat state variant of the same scheme requires 39 two-qubit gates, 20 simultaneous qubits, and achieves a circuit depth of only 1.
Moreover, a syndrome extraction scheme for codes of distance $d = 9$ is presented in \cite{cha20:FLagAllStabCodes}, where the measurement of a stabilizer operator with weight $\gamma_i = 18$ requires 306 two-qubit gates and 10 simultaneous qubits. In comparison, the proposed cut-cat state protocol, using the $t$-\ac{FT} cat-state preparation method from \cite{ForAma25:FTCSS}, requires only 70 two-qubit gates. The number of simultaneous qubits needed in this case is 13, and the circuit depth is 20.
A summarized comparison of $t$-\ac{FT} syndrome extraction schemes based on flag qubits, full cat states, and cut-cat states is presented in Table~\ref{tab:Resources}.

From this analysis, we observe that the proposed cut-cat scheme consistently requires fewer than half the number of simultaneous qubits compared to full cat state syndrome extraction, thereby enabling operation in qubit-constrained architectures with only a limited increase in circuit depth.
Moreover, the circuit depth, in the section involving interaction with data qubits, is significantly lower than in flag-based implementations, which reduces the accumulation of memory noise during the extraction round. 
Finally, the number of two-qubit gates is comparable for codes of distance $d = 5$ and $d = 7$ to the highly optimized schemes proposed in \cite{Pra23:OtimizedSchemeFlag5and7}.
However, for codes of distance $d = 9$, the proposed technique achieves a notable reduction in gate count compared to the standard approach presented in \cite{cha20:FLagAllStabCodes}, thereby reducing the number of two-qubit depolarizing noise insertions in circuit-level noise models.
When integrated into a complete \ac{FT} protocol, these features collectively reduce the contribution of syndrome extraction to the overall logical error rate.

\begin{table}[t]
    \centering
    \setlength{\tabcolsep}{3pt}
    \caption{Comparison between FT Syndrome Extraction Strategies}
    \label{tab:Resources}
    \small
    \resizebox{0.48\textwidth}{!}{
    \includegraphics[width=\textwidth]{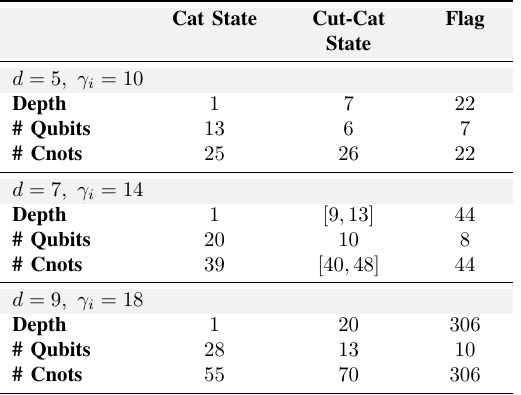}}
\end{table}


\section{Conclusions}\label{sec:conclusions}

In this work, we introduce a novel \ac{FT} syndrome extraction scheme based on cat states for CSS codes, with explicit  constructions demonstrated up to distance nine.
The scheme measures the cat-state stabilizers following their interaction with the data qubits, effectively constituting a second stage of error correction on the syndrome qubits.
This can be seen as a second stage of error correction on the syndrome qubits.
This approach preserves the advantages of cat state syndrome extraction, namely, parallelized interaction with data qubits, while reducing the required number of qubits for syndrome extraction by more than half.
Our analysis shows that the proposed method achieves a comparable two-qubit gate count to state-of-the-art flag qubit protocols for distances $d \leq 7$, while demonstrating a significant advantage for codes of distance $d = 9$, particularly as the weight of the measured operator increases.
Extensions to higher distances are left for future investigation.

\appendices

\section{Lists of Subsets $\mathcal{J}_{n_\mathrm{a}, k}$}
\label{app:Lists7}
In the following, we provide lists of subsets of cat stabilizer indices  $\mathcal{J}_{n_\mathrm{a}, k}$.
These lists are defined $\forall j = 0, \dots, \gamma_i/2 - 1$.
\begingroup
\allowdisplaybreaks
\begin{align}
    \notag {J}_{1, 0}(j) &= \left\{ j + 1, j + 2 \right\} \\
    \notag {J}_{1, 1}(j) &= \left\{ j - 3 , j -2\right\} \\
    \notag {J}_{2,0}(j) &= \left\{ j + 2, j + 3, j + 4 \right\} \\
    \notag {J}_{2,1}(j) &= \left\{ j -1, j, j + 1 \right\} \\
    \notag {J}_{2,2}(j) &= \left\{ j + 2, j + 3 \right\} \\
    \notag {J}_{2,3}(j) &= \left\{ j -2, j -1 \right\} \\
    \notag {J}_{2,4}(j) &= \left\{ j - 1, j + 1 \right\} \\
    \notag {J}_{3,0}(j) &= \left\{ j + 3, j + 4 \right\} \\
    \notag {J}_{3,1}(j) &= \left\{ j + 2, j + 3 \right\} \\
    \notag {J}_{3,2}(j) &= \left\{ j - 1, j\right\} \\
    \notag {J}_{4,0}(j) &= \left\{ j + 2, j + 3\right\} \\
    \notag {J}_{4,1}(j) &= \left\{ j - 1, j \right\} \\
    \notag {J}_{4,2}(j) &= \left\{ j - 1, j \right\} \\
    \notag {J}_{4,3}(j) &= \left\{ j + 3, j + 4\right\} \\
    \notag {J}_{5,0}(j) &= \left\{ j + 3, j + 4\right\}.
\end{align}
\endgroup

These lists are used to compare outcomes of subsets of cat stabilizers across different measurement rounds, in order to minimize their number.

\section{Cut-Cat LUT for Syndrome Extraction}
\label{app:LUT}

In this appendix, we present the pseudocode used to generate the \ac{LUT} for cut-cat syndrome extraction at a given code distance. 
The pseudocode is given in Algorithm~\ref{alg:cutcat_lut}.
In particular, we consider the following fault locations: $(i)$ errors occurring before the start of the circuit on each cat qubit,
$(ii)$ errors on the control qubit of each two-qubit gate, occurring after the gate's execution, and
$(iii)$ errors occurring just before each measurement.
In the algorithm, $\texttt{wt}()$ denotes the function that computes the Hamming weight of a binary vector.
The algorithm takes as input the stabilizer weight $\gamma_i$ of the generator being measured and the error correction capability of the code, $t$.
Upon successful execution, it returns a \ac{LUT} that associates each cat syndrome with a binary correction vector $\V{c}_s$ of length $\gamma_i$, which specifies the correction to be applied to the propagated data error.
At run time, the decoder performs an $O(1)$ expected-time LUT lookup on the measured syndrome and applies the corresponding correction, for a worst-case decoding time $O(\gamma_i)$.

\begin{algorithm}[t]
\caption{Cut-cat LUT construction}
\label{alg:cutcat_lut}
\KwIn{Stabilizer weight $\gamma_i$, code capability $t$}
\KwOut{LUT mapping cat syndromes to corrections, or \textsc{Fail}}
Initialize an empty LUT; \\
Initialize $\mathcal{E}_s = \{\}$ for each syndrome $\V{s}$; \\
\ForAll{$w \in \{1,\dots, t$\}}{
  \ForEach{combination of $w$ fault locations}{
    Insert $\M{X}$ errors at the chosen locations and propagate them through the syndrome extraction circuit; \\
    Let $\V{e} \in \{0,1\}^{\gamma_i}$ be the resulting data error pattern and let $\V{s}$ be the resulting cut-cat syndrome; \\
    Insert $(\V{e},w)$ into the set $\mathcal{E}_s$; \\
    Search for a correction vector $\V{c}_s \in \{0,1\}^{\gamma_i}$ such that
    $\texttt{wt}(\V{e}' \oplus \V{c}_s) \le w'$ for all $(\V{e}',w') \in \mathcal{E}_s$; \\
    \If{no such correction vector exists}{
      \Return{\textsc{Fail}}; \\
    }
    Store $\V{c}_s$ for syndrome $\V{s}$ in the LUT;
  }
}
\Return{\textsc{Success} and the completed LUT};
\end{algorithm}

\section{Proof of Theorem~\ref{th:FTdist7} }
\label{app:FTdist7}
When measuring a code generator of weight $\M{\gamma_i}$, an error of weight $\leq 3$ may cause up to six cat stabilizer measurements to trigger. 
First, note that both a measurement error on $M_{r}^{(i)}$ and an error on the cat qubit $r$ occurring between the two-qubit gates that interact with it during the $i$-th cat stabilizer measurement trigger the same cat stabilizer.
The key difference between these two error types is that the former affects only the first round of cat stabilizer measurements, whereas the latter results in a persistent parity flip in that cat stabilizer across all subsequent rounds. 
Nevertheless, both errors can be detected in a similar way using additional rounds of cat stabilizer measurements.
For clarity and conciseness, we will consider only cat measurement errors, though the discussion applies equally to both types of error.
In the following, we describe how the decoding steps to be applied based on $|\mathcal{A}|$, achieve fault-tolerance. 

\begin{itemize}
    \item If a trivial syndrome is observed after the first round, this indicates either that no errors have occurred, that two errors have occurred on the same cat qubit, or that a single cat qubit error was followed by two measurement errors. 
    In any of these cases, no error correction is needed.
    \item A single cat state stabilizer can be triggered by either a measurement error on its own, or a combination of one or two adjacent cat qubit errors and a measurement error.
    In the last case, we need to apply a correction, since an error of weight $w = 3$ can be propagated to an error of weight $w = 4$.
    To address this, such measurement errors are detected by checking the consistency of outcomes for selectively chosen stabilizers in a second round of measurements.
    For instance, if $M_{r}^{(0)}$ is the triggered cat stabilizer, it is sufficient to check wether one between $M_{r - 2}^{(1)}$ and $M_{r + 2}^{(1)}$ is turned on, and apply a correction on data qubit $q_{2r+1}$ or $q_{2r+3}$, respectively. 
    \item Two cat stabilizers can be triggered by a combination of measurement errors and one or more adjacent cat qubit errors.
    Specifically, the instances with either two or three adjacent cat qubit errors must be corrected, since errors of weight $w = 2$ and $w = 3$ are propagated to errors of weight up to $w = 4$ and $w = 6$, respectively.
    In the case where the triggered pair is separated by three cat qubits, it is sufficient to measure again either one of them and, if the outcome is consistent, apply a correction to three of the data qubits interacting with the cat qubits between the pair.
    Instances where the triggered pair is separated by only two cat qubits require more involved checks.
    Specifically, it is necessary to distinguish between: (i) two adjacent cat qubit errors, possibly accompanied by a third measurement error in the second round; (ii) a single cat qubit error combined with measurement errors; and (iii) a pair of faulty measurements.
    For instance, if $M_{r}^{(0)}$ and $M_{r + 2}^{(0)}$ are triggered, we first measure $M_{r}^{(1)}$.
    If this measurement is off, we proceed to measure $M_{r + 2}^{(1)}$ and $M_{r + 3}^{(1)}$. In this case, no correction is applied if either $M_{r + 2}^{(1)}$ is off or $M_{r + 3}^{(1)}$ is on.
    Indeed, if both $M_{r}^{(1)}$ and $M_{r + 2}^{(1)}$ are off, the error pattern is consistent with two faulty measurements $M_{r}^{(0)}$ and $M_{r + 2}^{(0)}$.
    If $M_{r + 2}^{(1)}$ and $M_{r + 3}^{(1)}$ are on, the error consists of two faulty measurements ($M_{r}^{(0)}$ and $M_{r + 3}^{(0)}$) and a cat error on the $(r+3)$-th qubit.
    Conversely, if $M_{r}^{(1)}$ is on, we measure $M_{r + 2}^{(1)}$ and $M_{r - 1}^{(1)}$.
    If $M_{r + 2}^{(1)}$ is off and $M_{r - 1}^{(1)}$ is on, no correction is applied, as this pattern corresponds to two faulty measurements $M_{r- 1}^{(0)}$ and $M_{r + 2}^{(0)}$ combined with a single cat error on the $r$-th qubit.
    In all the other cases it is necessary to apply a correction to two of the data qubits interacting with the cat qubits between the pair.
    \item  Three cat stabilizers can be triggered by either a single cat qubit error combined with one measurement error, or by two cat qubit errors along with a measurement error.
    In the case where $M_{r}^{(0)}$, $M_{r + 2}^{(0)}$, and $M_{r + 4}^{(0)}$ are triggered, it is sufficient to measure $M_{r + 4}^{(1)}$.
    If the outcome is on, apply a correction to one of the data qubits interacting with the cat qubits between the $(r + 2)$-th and $(r + 4)$-th cat stabilizers, e.g., $q_{2(r+3)+3}$.
    Conversely, if the additional measurement is off, a correction should be applied to one of the data qubits interacting with the cat qubits between the $r$-th and $(r + 2)$-th cat stabilizers, e.g., $q_{2r+3}$.
    Cases where $M_{r}^{(0)}$, $M_{r + 1}^{(0)}$, and $M_{r + 2}^{(0)}$ are triggered can be caused either by two cat qubit errors plus a measurement fault, or by two measurement faults and a single cat qubit error.
    In this case, it is necessary to measure $M_{r - 1}^{(1)}$.
    If this outcome is off, the error may consist of two adjacent cat errors on the $(r + 1)$-th and $(r + 2)$-th qubits along with a measurement fault $M_{r + 1}^{(0)}$.
    We then apply a correction on $q_{2(r+1)+1}$.    
    If instead $M_{r - 1}^{(1)}$ is on, we proceed to measure $M_{r}^{(1)}$.
    If this outcome is on, the error may be due to cat qubit errors on the $r$-th and $(r + 2)$-th qubits, together with a measurement fault $M_{r - 1}^{(0)}$, and we apply a correction on $q_{2(r+1)+3}$.
    Finally, if $M_{r}^{(1)}$ is off, we measure $M_{r + 1}^{(1)}$.
    If this result is on, the error may consist of two cat errors on the $(r + 1)$-th and $(r + 3)$-th qubits plus a measurement fault on $M_{r + 3}^{(0)}$, and we apply a correction on $q_{2(r+1)+1}$.
    In the case where $M_{r}^{(0)}$, $M_{r + 2}^{(0)}$, and $M_{r + 3}^{(0)}$ are triggered, it is sufficient to measure $M_{r + 3}^{(1)}$.
    If this measurement is off, the error may consist of two cat errors on the $(r + 1)$-th and $(r + 2)$-th qubits and a measurement fault ($M_{r + 3}^{(0)}$).
    We then apply a correction on $q_{2(r+1)+1}$.
    An analogous situation occurs if $M_{r}^{(0)}$, $M_{r + 1}^{(0)}$, and $M_{r + 3}^{(0)}$ are triggered.
    In this case, we measure $M_{r}^{(1)}$, and if the result is off, the error may consist of two cat errors on the $(r + 2)$-th and $(r + 3)$-th qubits along with a measurement fault ($M_{r}^{(0)}$).
    Accordingly, we apply a correction on $q_{2(r+1)+3}$.
    There are only two remaining error configurations that need to be addressed.
    The first case involves two cat qubit errors occurring on the $r$-th and the $u$-th cat qubits, where $u > r + 2$, along with a measurement fault.
    The second case consists of two adjacent cat qubit errors and an additional measurement fault on $M_{u}^{(0)}$, such that if $M_{r}^{(0)}$ and $M_{r + 2}^{(0)}$ are triggered, then $u > r + 2$.
    In the first case, if none of the previous corrections have been applied, we check whether two consecutive measurements, $M_{r}^{(0)}$ and $M_{r + 1}^{(0)}$, are triggered.
    If so, we apply a correction on $q_{2(r+1)+3}$.
    In the second case, if none of the previous corrections have been applied, we check whether two  measurements, $M_{r}^{(0)}$ and $M_{r + 2}^{(0)}$, are triggered.
    If so, we apply a correction on $q_{2(r+1)+3}$.
    \item Four cat stabilizers can be triggered by various combinations of errors: two or three cat qubit errors, or one or two cat qubit errors combined with two faulty measurements.
    In particular, configurations involving more than one cat qubit error must be corrected.
    If $M_{r}^{(0)}$, $M_{r + 2}^{(0)}$, $M_{r + 3}^{(0)}$, and $M_{r + 4}^{(0)}$ are triggered, we measure $M_{r + 4}^{(1)}$.
    If this result is on, no correction is applied.
    This pattern likely results from a single cat  error on the $(r + 3)$-th qubit and two measurement faults $M_{r}^{(0)}$ and $M_{r + 4}^{(0)}$.
    In another instance, when $M_{r}^{(0)}$, $M_{r + 2}^{(0)}$, $M_{r + 3}^{(0)}$, and $M_{r + 4}^{(0)}$ are triggered, we measure both $M_{r}^{(1)}$ and $M_{r + 3}^{(1)}$.
    If both outcomes are off, this pattern may be caused by a cat qubit error on the $(r + 2)$-th qubit and two measurement faults $M_{r}^{(0)}$ and $M_{r + 3}^{(0)}$.
    Again, no correction is applied.
    Finally, if $M_{r}^{(0)}$, $M_{r + 1}^{(0)}$, $M_{r + 2}^{(0)}$, and $M_{r + 4}^{(0)}$ are triggered, we measure $M_{r}^{(1)}$.
    If this measurement is off, the error may consist of a cat qubit error on the $(r + 2)$-th qubit and two faulty measurements $M_{r}^{(0)}$ and $M_{r + 4}^{(0)}$.
    In this case as well, no correction is applied.
    For all the other error patterns it is safe to apply a correction, as shown in Algorithm~\ref{algo:D7_4}.
    Recall that the number of cat qubits must be at least equal to the code distance, so that any cat qubit error of weight $\leq 3$ can always be correctly identified by the cat stabilizer measurements.
    \item Five cat stabilizers can be triggered by two cat qubit errors and a faulty measurement.
    In case $M_{r}^{(0)}$, $M_{r + 1}^{(0)}$, $M_{r + 2}^{(0)}$, $M_{r + 3}^{(0)}$, and $M_{r + 4}^{(0)}$ are triggered, we measure $M_{r + 4}^{(1)}$.
    If $M_{r + 4}^{(1)}$ is on, the error corresponds to two cat qubit errors on the $(r + 2)$-th and $(r + 4)$-th qubits, along with a faulty $M_{r}^{(0)}$ measurement. In this case, we apply a correction on $q_{2(r+3)+3}$.
    On the other hand, if $M_{r + 4}^{(1)}$ is off, the error consists of two cat qubit errors on the $(r + 1)$-th and $(r + 3)$-th qubits, and a faulty $M_{r + 4}^{(0)}$ measurement. Accordingly, we correct $q_{2(r+2)+3}$.
    On the other hand, if the triggered cat stabilizers are not all adjacent to each other, two additional configurations must be considered.
    In the case where two adjacent stabilizers and three additional adjacent stabilizers are triggered—specifically, the $r$-th, $(r+1)$-th, $(r+2)$-th, $u$-th, and $(u+1)$-th stabilizers, with $u > r + 2$—it is sufficient to apply a correction on $q_{2(u + 1)+1}$.
    Furthermore, if four adjacent stabilizers are triggered—namely, the $r$-th, $(r+1)$-th, $(r+2)$-th, $(r+3)$-th, and $u$-th stabilizers, with $u > r + 4$—then it is sufficient to apply a correction on $q_{2(r + 1)+1}$.
    In both cases, one of the data qubit errors resulting from the cat qubit errors is always corrected, effectively reducing the worst-case scenario from a weight-$w = 4$ data qubit error to a weight-$w = 3$ error. 
    Since this is due to three errors, the scheme is $3$-\ac{FT}.
\end{itemize}




\bibliographystyle{IEEEtran}
\bibliography{Files/IEEEabrv,Files/StringDefinitions,Files/StringDefinitions2,Files/refs}

@STRING{IEEE_O_CSTO       = "{IEEE} Commun. Surveys Tuts."}

@STRING{IEEE = {The Institute of Electrical and Electronics Engineers}}

@STRING{PR-A = {Phys. Rev. A}}

@STRING{PRL = {Phys. Rev. Lett.}}

@article{DenKitLan:02,
	author = {Eric Dennis and Alexei Kitaev and Andrew Landahl and John Preskill},
	journal = {Journal of Mathematical Physics},
	month = {sep},
	number = {9},
	pages = {4452--4505},
	title = {Topological quantum memory},
	volume = {43},
	year = 2002}

@article{FowMarMar:12,
	author = {Austin G. Fowler and Matteo Mariantoni and John M. Martinis and Andrew N. Cleland},
	journal = PR-A,
	month = {sep},
	number = {3},
	title = {Surface codes: Towards practical large-scale quantum computation},
	volume = {86},
	year = 2012}

@book{NieChu:10,
	author = {Nielsen, Michael A. and Chuang, Isaac L.},
	publisher = {Cambridge University Press},
	title = {Quantum Computation and Quantum Information},
	year = {2010}}

@article{Sho:95,
	author = {Shor, Peter W.},
	journal = {Phys. Rev. A},
	month = {Oct},
	pages = {R2493--R2496},
	title = {Scheme for reducing decoherence in quantum computer memory},
	volume = {52},
	year = {1995}}

@Article{Got:96,
  title = {Class of quantum error-correcting codes saturating the quantum {Hamming} bound},
  author = {Gottesman, Daniel},
  journal = PR-A,
  volume = {54},
  issue = {3},
  pages = {1862--1868},
  numpages = {0},
  year = {1996},
  month = {Sep},
  publisher = {American Physical Society},
  doi = {10.1103/PhysRevA.54.1862},
}

@inproceedings{Got:09,
  title={An introduction to quantum error correction and fault-tolerant quantum computation},
  author={Gottesman, Daniel},
  booktitle={Quantum information science and its contributions to mathematics, {PSAM}},
  volume={68},
  pages={13--58},
  year={2010}
}

@article{Laf:96,
 	title={Perfect quantum error correcting code},
	 author={Laflamme, Raymond and Miquel, Cesar and Paz, Juan Pablo and Zurek, Wojciech Hubert},
	 journal=PRL,
	 volume={77},
	 number={1},
	 pages={198},
	 year={1996},
	 publisher={APS}
}

@article{Kni:97,
	 title = {Theory of quantum error-correcting codes},
	 author = {Knill, Emanuel and Laflamme, Raymond},
	 journal = {Phys. Rev. A},
	 volume = {55},
	 issue = {2},
	 pages = {900--911},
	 numpages = {0},
	 year = {1997},
	 month = {Feb},
	 publisher = {American Physical Society},
	 doi = {10.1103/PhysRevA.55.900},
}

@ARTICLE{Bab:19,
	author={Z. {Babar} and D. {Chandra} and H. V. {Nguyen} and P. {Botsinis} and D. {Alanis} and S. X. {Ng} and L. {Hanzo}},
	journal=IEEE_O_CSTO,
	title={Duality of Quantum and Classical Error Correction Codes: Design Principles and Examples},
	year={2019},
	volume={21},
	number={1},
	pages={970-1010},
	doi={10.1109/COMST.2018.2861361},
	ISSN={1553-877X},
	month={Firstquarter}
}

@article{FleShoWin:08,
	author = {Fletcher, Andrew S. and Shor, Peter W. and Win, Moe Z.},
	journal = {Phys. Rev. A},
	month = {Jan},
	pages = {012320},
	title = {Structured near-optimal channel-adapted quantum error correction},
	volume = {77},
	year = {2008}}

@article{Ter:15,
	author = {Terhal, Barbara M.},
	journal = {Rev. Mod. Phys.},
	month = {Apr},
	pages = {307--346},
	title = {Quantum error correction for quantum memories},
	volume = {87},
	year = {2015}}

@article{Del:21,
  title={Almost-linear time decoding algorithm for topological codes},
  author={Delfosse, Nicolas and Nickerson, Naomi H},
  journal={Quantum},
  volume={5},
  pages={595},
  year={2021},
  publisher={Verein zur F{\"o}rderung des Open Access Publizierens in den Quantenwissenschaften}
}

@article{ForValChi:23,
  title={Logical Error Rates of {XZZX} and Rotated Quantum Surface Codes},
  author={Forlivesi, Diego and Valentini, Lorenzo and Chiani, Marco},
  journal={IEEE JSAC},
  year={2024}
}

@article{Hig:23,
	author = {Higgott, Oscar and Bohdanowicz, Thomas C. and Kubica, Aleksander and Flammia, Steven T. and Campbell, Earl T.},
	journal = {Phys. Rev. X},
	month = {Jul},
	pages = {031007},
	title = {Improved Decoding of Circuit Noise and Fragile Boundaries of Tailored Surface Codes},
	volume = {13},
	year = {2023} }

@article{ForValChi24:MacW,
  title={Performance Analysis of Quantum Error-Correcting Codes via {Williams} Identities},
  author={Diego Forlivesi and Lorenzo Valentini and Marco Chiani},
  journal={Quantum},
  volume={9},
  pages={1950},
  year={2025},
  publisher={Verein zur F{\"o}rderung des Open Access Publizierens in den Quantenwissenschaften}
}

@article{HigGid:23,
  doi = {10.22331/q-2025-01-20-1600},
  title = {Sparse {B}lossom: correcting a million errors per core second with minimum-weight matching},
  author = {Higgott, Oscar and Gidney, Craig},
  journal = {{Quantum}},
  issn = {2521-327X},
  publisher = {{Verein zur F{\"{o}}rderung des Open Access Publizierens in den Quantenwissenschaften}},
  volume = {9},
  pages = {1600},
  month = jan,
  year = {2025}
}

@article{CalSho:96,
  title={Good quantum error-correcting codes exist},
  author={Calderbank, A Robert and Shor, Peter W},
  journal=PR-A,
  volume={54},
  number={2},
  pages={1098},
  year={1996},
  publisher={APS},
  doi = {https://doi.org/10.1103/PhysRevA.54.1098}
}

@article{Ste:96,
  title={Multiple-particle interference and quantum error correction},
  author={Steane, Andrew},
  journal={Proceedings of the Royal Society of London. Series A: Mathematical, Physical and Engineering Sciences},
  volume={452},
  number={1954},
  pages={2551--2577},
  year={1996},
  publisher={The Royal Society London},
 doi = {https://doi.org/10.1098/rspa.1996.0136}
}

@article{Bom06:colorCodes,
  title={Topological quantum distillation},
  author={Bombin, Hector and Martin-Delgado, Miguel Angel},
  journal=PRL,
  volume={97},
  number={18},
  pages={180501},
  year={2006},
  publisher={APS}
}

@ARTICLE{ValForChi25:CylMob,
  author={Valentini, Lorenzo and Forlivesi, Diego and Chiani, Marco},
  journal={IEEE Trans. Inf. Theory.}, 
  title={Cylindrical and {M\"obius} Quantum Codes for Asymmetric {Pauli} Errors}, 
  year={2025},
  volume={71},
  number={5},
  pages={3766-3778},
  doi={10.1109/TIT.2025.3546769}
}

@inproceedings{Sho96:CatState,
  title={Fault-tolerant quantum computation},
  author={Shor, Peter W},
  booktitle={Proceedings of 37th conference on foundations of computer science},
  pages={56--65},
  year={1996},
  organization={IEEE}
}

@article{Ste97:SteaneGadget,
   title={Active Stabilization, Quantum Computation, and Quantum State Synthesis},
   volume={78},
   ISSN={1079-7114},
   DOI={10.1103/physrevlett.78.2252},
   number={11},
   journal={Physical Review Letters},
   publisher={American Physical Society (APS)},
   author={Steane, A. M.},
   year={1997},
   month=mar, pages={2252–2255} }

@article{cha20:FLagAllStabCodes,
  title={Flag fault-tolerant error correction for any stabilizer code},
  author={Chao, Rui and Reichardt, Ben W},
  journal={PRX Quantum},
  volume={1},
  number={1},
  pages={010302},
  year={2020},
  publisher={APS}
}

@article{Pra23:OtimizedSchemeFlag5and7,
  doi = {10.22331/q-2023-10-24-1154},
  title = {Fault-tolerant syndrome extraction and cat state preparation with fewer qubits},
  author = {Prabhu, Prithviraj and Reichardt, Ben W.},
  journal = {{Quantum}},
  issn = {2521-327X},
  publisher = {{Verein zur F{\"{o}}rderung des Open Access Publizierens in den Quantenwissenschaften}},
  volume = {7},
  pages = {1154},
  month = oct,
  year = {2023}
}

@article{ForValChi25:Bub,
  title={Bubble Clustering Decoder for Quantum Topological Codes},
  author={Forlivesi, Diego and Valentini, Lorenzo and Chiani, Marco},
  journal={IEEE Trans. on Comm.},
  year={2025},
  publisher={IEEE}
}

@article{cha18:FirstFLagFT,
  title={Fault-tolerant quantum computation with few qubits},
  author={Chao, Rui and Reichardt, Ben W},
  journal={npj Quantum Information},
  volume={4},
  number={1},
  pages={42},
  year={2018},
  publisher={Nature Publishing Group UK London}
}

@article{Ste:14,
      title={Efficient fault-tolerant decoding of topological color codes}, 
      author={Ashley M. Stephens},
      year={2014},
      eprint={1402.3037},
      journal={arXiv preprint quant-ph/1402.3037},
      archivePrefix={arXiv},
      primaryClass={quant-ph}
}

@article{Yod:17,
  title={The surface code with a twist},
  author={Yoder, Theodore J and Kim, Isaac H},
  journal={Quantum},
  volume={1},
  pages={2},
  year={2017},
  publisher={Verein zur F{\"o}rderung des Open Access Publizierens in den Quantenwissenschaften}
}

@article{cha19:MagicStateFlags,
  title={Fault-tolerant magic state preparation with flag qubits},
  author={Chamberland, Christopher and Cross, Andrew W},
  journal={Quantum},
  volume={3},
  pages={143},
  year={2019},
  publisher={Verein zur F{\"o}rderung des Open Access Publizierens in den Quantenwissenschaften}
}

@article{Cha18:FlagsDet,
   title={Flag fault-tolerant error correction with arbitrary distance codes},
   volume={2},
   ISSN={2521-327X},
   url={http://dx.doi.org/10.22331/q-2018-02-08-53},
   DOI={10.22331/q-2018-02-08-53},
   journal={Quantum},
   publisher={Verein zur Forderung des Open Access Publizierens in den Quantenwissenschaften},
   author={Chamberland, Christopher and Beverland, Michael E.},
   year={2018},
   month=feb, pages={53} }

@article{rof20:BPplusOSD,
  title={Decoding across the quantum low-density parity-check code landscape},
  author={Roffe, Joschka and White, David R and Burton, Simon and Campbell, Earl},
  journal={Physical Review Research},
  volume={2},
  number={4},
  pages={043423},
  year={2020},
  publisher={APS}
}

@inproceedings{ValForChi25:latency,
  title={Impact of Decoding Latency in the Assessment of Quantum Surface Codes Performance},
  author={Valentini, Lorenzo and Forlivesi, Diego and Chiani, Marco},
  booktitle={2025 International Conference on Quantum Communications, Networking, and Computing (QCNC)},
  pages={554--559},
  year={2025},
  organization={IEEE}
}

@article{ForAma25:FTCSS,
  title={Flag at origin: a modular fault-tolerant preparation for {CSS} codes},
  author={Forlivesi, Diego and Amaro, David},
  journal={arXiv preprint arXiv:2508.14200},
  year={2025}
}

@article{TransversalClifford25:ShuVic,
  title={Transversal Clifford and {T}-gate codes of short length and high distance},
  author={Jain, Shubham P and Albert, Victor V},
  journal={IEEE Journal on Selected Areas in Information Theory},
  year={2025},
  publisher={IEEE}
}

@article{TanBal23:t+1RepeatedSyndromeExraction,
  title={Adaptive syndrome measurements for Shor-style error correction},
  author={Tansuwannont, Theerapat and Pato, Balint and Brown, Kenneth R},
  journal={Quantum},
  volume={7},
  pages={1075},
  year={2023},
  publisher={Verein zur F{\"o}rderung des Open Access Publizierens in den Quantenwissenschaften}
}

@article{Breebe21:QLDPCcodes,
  title={Quantum low-density parity-check codes},
  author={Breuckmann, Nikolas P and Eberhardt, Jens Niklas},
  journal={PRX quantum},
  volume={2},
  number={4},
  pages={040101},
  year={2021},
  publisher={APS}
}

@article{BraHaa:12,
  title={Magic-state distillation with low overhead},
  author={Bravyi, Sergey and Haah, Jeongwan},
  journal={Physical Review A—Atomic, Molecular, and Optical Physics},
  volume={86},
  number={5},
  pages={052329},
  year={2012},
  publisher={APS}
}

@article{DagBluBro:25,
  title={Experimental demonstration of high-fidelity logical magic states from code switching},
  author={Daguerre, Lucas and Blume-Kohout, Robin and Brown, Natalie C and Hayes, David and Kim, Isaac H},
  journal={arXiv preprint arXiv:2506.14169},
  year={2025}
}

@article{Got24:ManyHypercube,
  title={Many-hypercube codes: High-rate quantum error-correcting codes for high-performance fault-tolerant quantum computing},
  author={Goto, Hayato},
  journal={arXiv preprint arXiv:2403.16054},
  year={2024}
}

@article{Yam24:ConcatenatedCodes,
  title={Time-efficient constant-space-overhead fault-tolerant quantum computation},
  author={Yamasaki, Hayata and Koashi, Masato},
  journal={Nature Physics},
  volume={20},
  number={2},
  pages={247--253},
  year={2024},
  publisher={Nature Publishing Group UK London}
}

@inproceedings{AhaBen97:FTschemeFp,
  title={Fault-tolerant quantum computation with constant error},
  author={Aharonov, Dorit and Ben-Or, Michael},
  booktitle={Proceedings of the twenty-ninth annual ACM symposium on Theory of computing},
  pages={176--188},
  year={1997}
}

\end{document}